\documentclass[12pt,draftcls,onecolumn]{IEEEtran}
\usepackage{amsmath}
\usepackage{amsthm}
\usepackage{amsfonts}
\usepackage{amssymb}
\usepackage[utf8]{inputenc}
\DeclareUnicodeCharacter{00A0}{ }
\usepackage{graphicx}
\usepackage{graphics}
\usepackage{float}
\usepackage{tikz}
\usetikzlibrary{shapes,snakes,patterns}
\usepackage{epstopdf}

\usepackage[justification=centering]{caption}
\usepackage{enumerate} 
\usepackage{cite}
\usepackage{placeins}
\newcommand{\vast}{\bBigg@{4}}
\newcommand{\Vast}{\bBigg@{5}}
\newtheorem{theorem}{Theorem}[]
\newtheorem{lemma}[]{Lemma}
\usepackage{suffix}
\usepackage{mathtools}
\usepackage{amsthm}
\usepackage{enumitem}
\usepackage{xfrac}
\usepackage{relsize}
\theoremstyle{definition}

\usepackage[list=true]{subcaption}
\begin{document}
\bstctlcite{IEEEexample:BSTcontrol}
	\title{Cooperative relaying in a SWIPT network: Asymptotic analysis using extreme value theory for non-identically distributed RVs}
	\author{Athira Subhash, Sheetal Kalyani \\
		\hspace{-0.5 cm}Department of Electrical Engineering,\\
		\hspace{-0.8 cm} Indian Institute of Technology, Madras, \\
		\hspace{-1cm} Chennai, India 600036.\\
		\hspace{-1cm} \{ee16d027@smail,skalyani@ee\}.iitm.ac.in\\
	}
\maketitle
\begin{abstract}
This paper derives the distribution of the maximum end-to-end (e2e) signal to noise ratio (SNR) in an opportunistic relay selection based cooperative relaying (CR) network having large number of non-identical relay links between the source and destination node. The source node is assumed to be simultaneous wireless information and power transfer (SWIPT) enabled and the relays are capable of both time splitting (TS) and power splitting (PS) based energy harvesting (EH). Contrary to the majority of literature in communication, which uses extreme value theory (EVT) to derive the statistics of extremes of sequences of independent and identically distributed (i.i.d.) random variables (RVs), we demonstrate how tools from EVT can be used to derive the asymptotic statistics of sequences of independent and non-identically distributed (i.n.i.d.) normalised SNR RVs and {hence characterise the distribution of the maximum e2e SNR RV.}  Using these results we derive the expressions for ergodic and outage capacities at the destination node. Finally, we present the utility of the asymptotic results for deciding the optimum TS and PS factors of the hybrid EH relays that (i) minimise outage probability and (ii) maximise ergodic capacity at the destination. Furthermore, we demonstrate how stochastic ordering results can be utilised for simplifying these optimisation problems.
\end{abstract}


\section{Introduction}
\par Wireless power transfer (WPT) has been an active area of research since the early 20-th century \cite{garnica2013wireless}. Recent technological advancements have further triggered an interest in exploring viable applications of WPT in many diverse fields \cite{perera2017simultaneous}. Amongst these applications, the idea of simultaneous wireless information and power transfer (SWIPT) has received much focus both from the industry and academia recently \cite{perera2017simultaneous,krikidis2014simultaneous}. As the name suggests, SWIPT aims at transporting power using information signals and thus allows nodes in a wireless network to harvest energy from any incoming information signal. The idea of energy harvesting (EH) nodes in wireless communication brings out the possibility of self-sustaining nodes and is considered as a reliable and viable alternative to battery-powered wireless nodes. Practical receiver architectures for wireless information and power transfer are studied in \cite{zhou2013wireless}. The theoretical and practical aspects of SWIPT are studied for several applications such as cognitive radio \cite{ren2018rf}, wireless sensor networks \cite{liu2017rf}, personal area networks \cite{jiang2016social}, device to device communication \cite{yang2014green}, cooperative relaying (CR) \cite{liu2019cooperative}  and many more.
\par Recently, CR schemes with  EH nodes have gained significant attention \cite{tran2018rf,hossain2019survey}. CR has been identified as one of the promising technologies capable of addressing issues like fading, poor coverage, increased power consumption etc. \cite{gomez2011survey}. The benefits of CR come from the spatial diversity achieved by the virtual antenna array created by the cooperating nodes, which relay information between the source and the destination. Unlike co-located antenna arrays, the components of the virtual array have a dynamic nature depending upon the state and availability of the cooperating nodes. The authors of \cite{adinoyi2009performance,liu2019cooperative,gomez2011survey,zhang2017performance,liu2016wireless} and the references therein analyse different aspects of the CR schemes. These studies establish the merits offered by CR over other non-cooperative methods of communication. Self-sustainable CR can be realised with EH in the source or relay nodes of the CR network, thus reaping the benefits of both SWIPT and CR schemes together. Here, the source (or the relays) which has continuous access to an energy source transfer RF energy to the energy-starved relay (or the source).
\par Several recent works study the performance and resource allocation aspects of CR coupled with SWIPT nodes. The authors of \cite{hadzi2015rate,kumar2019performance} analyses the performance of dual-hop CR networks where the source is capable of harvesting energy from a single antenna relay node, whereas \cite{kamga2017relay} considers a multi-antenna EH relay node operating in millimetre wavebands. While \cite{kumar2019performance} analyses the average symbol error rate at the destination, \cite{kamga2017relay} studies the asymptotic (in terms of the number of relay antennas) energy harvested, spectral efficiency and system throughput at the destination. The secrecy performance and the optimal choices of system parameters to maximise rate of a multi-relay CR system with SWIPT enabled source nodes is studied in \cite{liu2016wireless} and \cite{xing2016wireless} respectively. Energy transfer in any SWIPT system follows one of the three EH protocols, namely, power splitting (PS) or time switching (TS) \cite{nasir2013relaying,zhang2013mimo} or a hybrid of the two \cite{atapattu2016optimal}. Energy received over a certain fraction of a time slot is used for EH and information processing (IP) is performed over the rest of the time slot in the TS protocol. Whereas in the PS protocol, a fraction of the received energy is used for EH and the rest for IP. The hybrid protocol performs both TS and PS in all the time slots. The choice of EH protocol depends on the system hardware constraints and in turn, decides the system performance. For example, in an EH system with TS protocol, the authors of \cite{hadzi2015rate,ghosh2019outage} discuss how to choose the optimal time fraction for EH such that rate maximisation and outage minimisation are respectively achieved. The authors of \cite{nguyen2018performance,zhong2018outage} study the performance of PS protocol in terms of outage probability, system throughput etc. Similarly, \cite{alsharoa2017optimization,malik2018optimizing} discuss the selection of optimal PS factor for different relay systems. Algorithms to arrive at the optimal PS factors in an multi relay assisted two hop CR communication network is studied in \cite{liu2016wireless}. Both the TS and PS protocols can be derived as special cases of the hybrid protocol and hence from an analysis perspective, the hybrid protocol is of particular interest. Therefore, in this work, we study the performance of a CR system with hybrid EH relay nodes harvesting energy from the information signal sent by the source node.
\par The first hop of a dual-hop CR communication involves information transmission from the source to the relays and then some or all of the relays forward this data to the destination in the second hop. Selecting a single relay for transmission in the second hop is an effective method to enhance the end-to-end (e2e) performance of relaying systems \cite{xia2015fundamental,li2017wireless} while keeping the decoding complexity at the destination minimal. Hence, several works like \cite{liu2019cooperative,medepally2010voluntary,bletsas2007cooperative,chen2010approximate} study the performance of CR where the relay with the largest e2e signal to interference plus noise ratio / signal to noise ratio (SNR) is chosen for information transmission to the destination. The asymptotic performance of the system when the number of relays grows to infinity facilitates easy comparison of the performance with respect to variations in other system parameters. For example, the authors of \cite{liu2019cooperative,medepally2010voluntary,song2006asymptotic,al2018secrecy} rely on asymptotic analysis for studying the system performance for diverse applications. Hence, in this work, we focus on the asymptotic e2e SNR of an EH-CR system where a large number of EH relays are available between the source and the destination node. Though the analysis is asymptotic, the results in Section \ref{Simulation} shows that they hold fairly well even in systems with a moderate number of relays between the source and the destination.
\par Extreme value theory (EVT) is a branch of statistics dealing with the asymptotics of extreme events (events with the extreme deviations from the median of probability distributions) \cite{charras2013extreme}. Tools from EVT has been efficiently used for solving several problems in wireless communication as well \cite{bennis2018ultrareliable,song2006asymptotic,haider2015spectral,al2018secrecy,ji2010capacity,liu2019cooperative,subhash2020transmit,kalyani2012asymptotic,subhash2019asymptotic,kalyani2012analysis}. Recently, \cite{liu2019cooperative} used EVT to analyse the asymptotic throughput of an opportunistic relay selection system when the relays are capable of harvesting energy from the desired signal as well as interferer signals. One important factor to notice in all these seminal works is that EVT has been used to derive the statistics of extremes over sequences of independent and identically distributed (i.i.d.) random variables (RVs). To the best of our knowledge, there is no previous work using EVT to derive the asymptotic statistics of extremes of a sequence of independent and non-identically distributed (i.n.i.d.) SNR RVs. While \cite{kalyani2012analysis} derives the pdf of the maximum of i.n.i.d. generalized-K variates using EVT, the pdf of each of the RV differ only by their mean values in this work. Hence, the sequence of i.n.i.d. RVs was easily transformed into a sequence of i.i.d. RVs by taking the difference of each RV with the common mean value. Thus, the analysis for i.i.d. RVs was used for the analysis of the extreme values of the sequence here. Several works like \cite{liu2019cooperative,xia2013spectrum,oyman2010opportunistic,xue2010} assume statistically identical source to relay and relay to destination links when analysing opportunistic relay selection schemes in CR models. However, note that each of the relay can be present at a different location with respect to the source node. Thus, the signal received at the different relays experience independent and non-identical path loss effects owing to the differences in path lengths. Similarly, the channel gain over each of the relay to destination links will also be i.n.i.d..Therefore it is imperative that one uses EVT and takes into account the i.n.i.d. nature of the relay links. {Hence, in this work, we use EVT to derive the distribution of the maximum e2e SNR in a dual-hop CR scenario with opportunistic relay selection and large number of non-identical links over the EH relay nodes.} We further highlight the need for a specific analysis of the statistics of the maximum of i.n.i.d. RVs using an example in Section.\ref{Simulation}.
\par Although the classical Fisher–Tippett theorem in EVT was proposed in the year 1928, the first work discussing the order statistics of sequences of i.n.i.d. RVs was published only 40 years later by Mejzler \cite{mejzler1969some}. The typical approach in identifying the asymptotic distribution of the maximum or the minimum over a sequence of normalised i.i.d. RVs using EVT includes the test to identify the maximum domain of attraction (MDA) \footnote{It is known that under certain conditions, the asymptotic distribution of the maximum or minimum RV of a sequence of normalised i.i.d. RVs $\{X_1,X_2,\cdots,X_n\}$ will only be one of the three extreme value distributions (EVD) (Frechet, Gumbel or Weibull). A distribution function $F$ is said to belong to the MDA of an EVD $G$ if $F^n(a_nx+b_n) = \mathbb{P}\left(\frac{\max\{X_1,\cdots,X_n\}-b_n}{a_n} \leq x \right) \to G(x)$ for some normalising constants $a_n$ and $b_n$.  } of the common distribution function and then finding the parameters of the asymptotic distribution (the normalising constants) \cite{falk2010laws}. The choice of these normalising constants are not unique and there are several common choices for all the possible asymptotic distributions available in literature \cite{de2007extreme,falk2010laws}. However, certain additional technical conditions are required to ensure the convergence of the distribution of the extreme statistic to a non-degenerate distribution function in the case of i.n.i.d. RVs. These conditions require the statistician to make appropriate choices for the normalising constants of the asymptotic distribution.  Although works like \cite{mejzler1969some,sreehari,barakat2002limit} presented conditions under which the asymptotic distribution of the maximum or the minimum of sequences of i.n.i.d. RVs exists, to the best of our knowledge, none of them provided general methods for identifying the specific normalising constants of the corresponding asymptotic distributions. Hence, the key challenge in characterising the distribution of the maximum e2e SNR over non-identical relay links is to identify the normalising constants of the asymptotic distribution. In this work, we derive one choice of normalising constants which enables us to characterise the asymptotic distribution of the maximum e2e SNR in a decode and forward (DF) CR system where the relays harvest energy from the source node via hybrid EH protocol. 

\subsection*{Outline}

In this work, we first derive the asymptotic distribution of the normalised maximum e2e SNR with i.n.i.d. source to destination links over dual-hop EH relays. {Using these results we characterise the distribution of the maximum e2e SNR RV.} We present the system model in Section \ref{sysmodel} and the detailed derivation of the asymptotic distribution is discussed in Section \ref{asymp_derive_sec}. The distribution of the e2e SNR is then used to derive the asymptotic ergodic and outage capacities at the destination in Section \ref{erg_cap_out_cap}.  {Next, in Section \ref{ordering_e2esnr} we use results from stochastic ordering to study the ordering of the maximum e2e SNR RV with respect to variations in different system parameters.} The asymptotic results are then used to identify the optimal TS and PS factor at the relays to (i) minimise the outage probability and  (ii) to maximise the ergodic capacity at the destination. Furthermore in Section \ref{opti_probsection}, we demonstrate the utility of the stochastic ordering results in simplifying and speeding up the solution for these optimisation problem. The validity of the results presented throughout the paper are demonstrated through simulation results in Section.\ref{Simulation} and finally, we conclude the work in Section \ref{conclusion}. \\
\textit{Notations:}The notations frequently used in this paper are summarised here. $\mathbb{E}[.]$ denotes expectation, $\mathbb{P}(Y)$ denotes probability of the event $Y$, $f_X(.)$ and $F_X(.)$ represent the probability distribution function (PDF) and cumulative distribution function (CDF) of random variable $X$, respectively. $e^{(.)}$ and $\exp(.)$ represent the exponential function, $E_n(.)$ indicate the exponential integral function \cite[Chapter 5]{iahandbook} and $\mathcal{CN}(\mu, \sigma^2)$ denotes complex Gaussian distribution with mean $\mu$ and variance $\sigma^2$.\\
\textit{Abbreviations:} Here, the frequently used abbreviations are presented in the \textit{Abbreviation-Expansion} format. WPT-Wireless Power Transfer, SWIPT-Simultaneous Wireless Information and Power Tranfer, EH-Energy Harvesting, CR-Cooperative Relaying, PS-Power Splitting, TS-Time Splitting, EVT-Extreme Value Theory, SNR-Signal to Noise Ratio, RV-Random Variable, DF-Decode and Forward, MDA-Maximum Domain of Attraction, i.i.d.-Independent and Identically Distributed, i.n.i.d.-Independent and Non-Identically Distributed and CDF-Cumulative Distribution Function.

\section{System Model} \label{sysmodel}

We consider a dual-hop CR scenario where a source node communicates with a destination node via energy-constrained relays equipped with EH circuitry. Here all the nodes are assumed to be equipped with a single antenna and capable of half duplex communication. The direct link between the source and the destination is assumed to be in permanent outage similar to \cite{liu2019cooperative,bjornson2013new,soliman2013dual}. The $L$ EH relays present between the source and the destination node harvest energy from the source and decode the data for the destination node. The relay which maximises the e2e SNR is then chosen to forward the data to the destination. Furthermore, similar to \cite{liu2019cooperative,kamga2017relay,hadzi2015rate}, we assume that our relays spend all the energy harvested from the source node during the EH phase of one time slot during the to send data to the destination over the information decoding phase. Fig \ref{sys_model} shows such a system model where $S$ represents the source node, $D$ the destination node and $\{R_1,\cdots,R_L\}$ are the $L$ EH relays. Here, $\{g_{1,\ell};\ell=1,\cdots,L\}$ and $\{g_{2,\ell};\ell=1,\cdots,L\}$ represent the small scale fading channel gains of the source to the $\ell^{th}$ relay and the $\ell^{th}$ relay to the destination links respectively. Similarly, $\{d_{1,\ell};\ell=1,\cdots,L\}$ and $\{d_{2,\ell};\ell=1,\cdots,L\}$ represent the distances from the source to the $\ell^{th}$ relay and the $\ell^{th}$ relay to the destination respectively. Furthermore, we assume that all the channels experience independent Rayleigh fading with $g_{i,\ell} \sim \mathcal{CN}(0,1); i\in \{1,2 \}$ and $\ell \in  \{1,\cdots,L \}$. Also, the channel is assumed to remain constant during the transmission of one block of information and it varies independently from one block to another. 
\begin{figure}[t]
	\centering
		\begin{minipage}[t]{0.4\textwidth}
		\centering
		\includegraphics[scale=0.3]{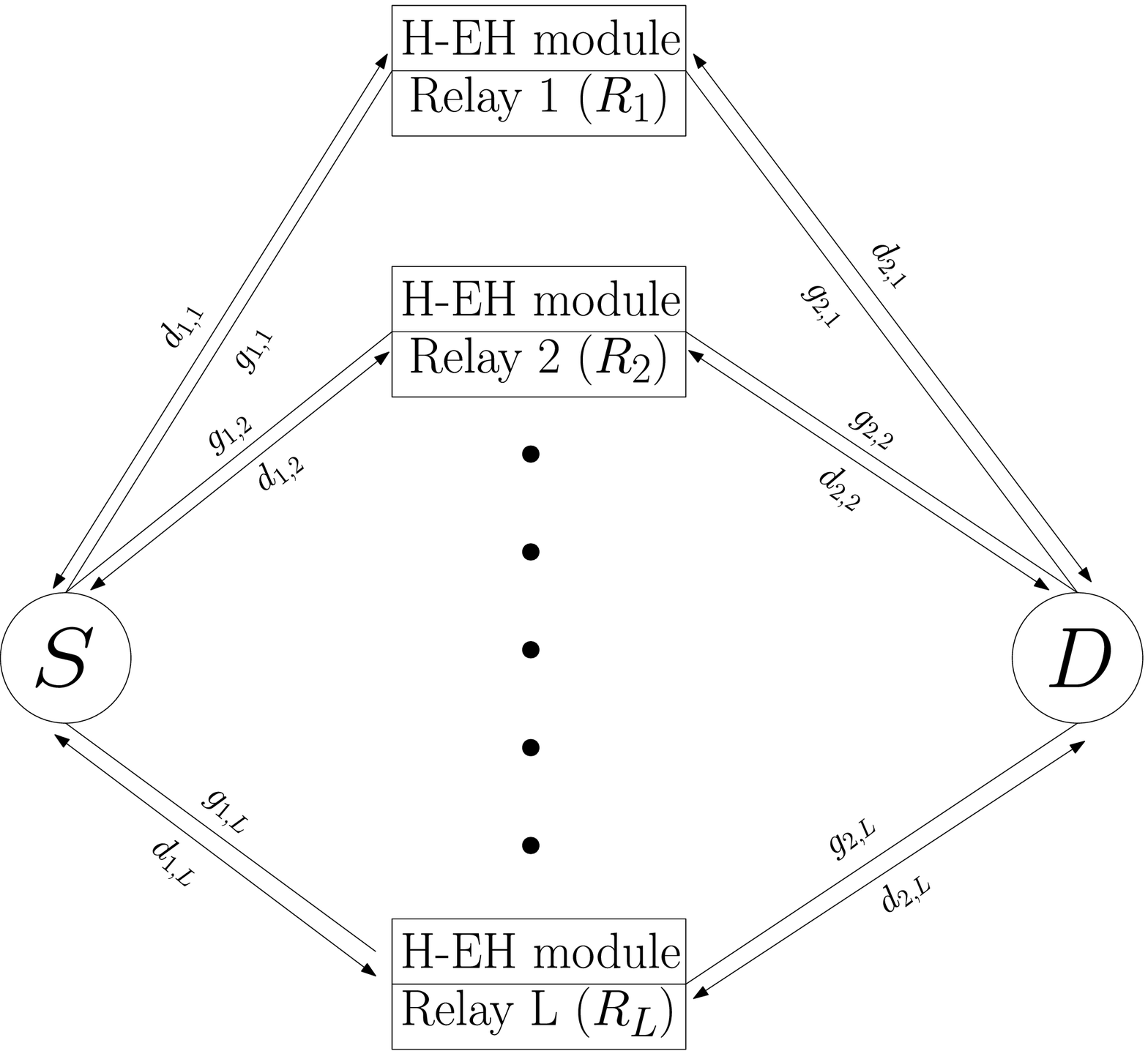}
		\caption{System model.}
		\label{sys_model}
	\end{minipage}
	\begin{minipage}[t]{0.5\textwidth}
		\includegraphics[scale=0.35]{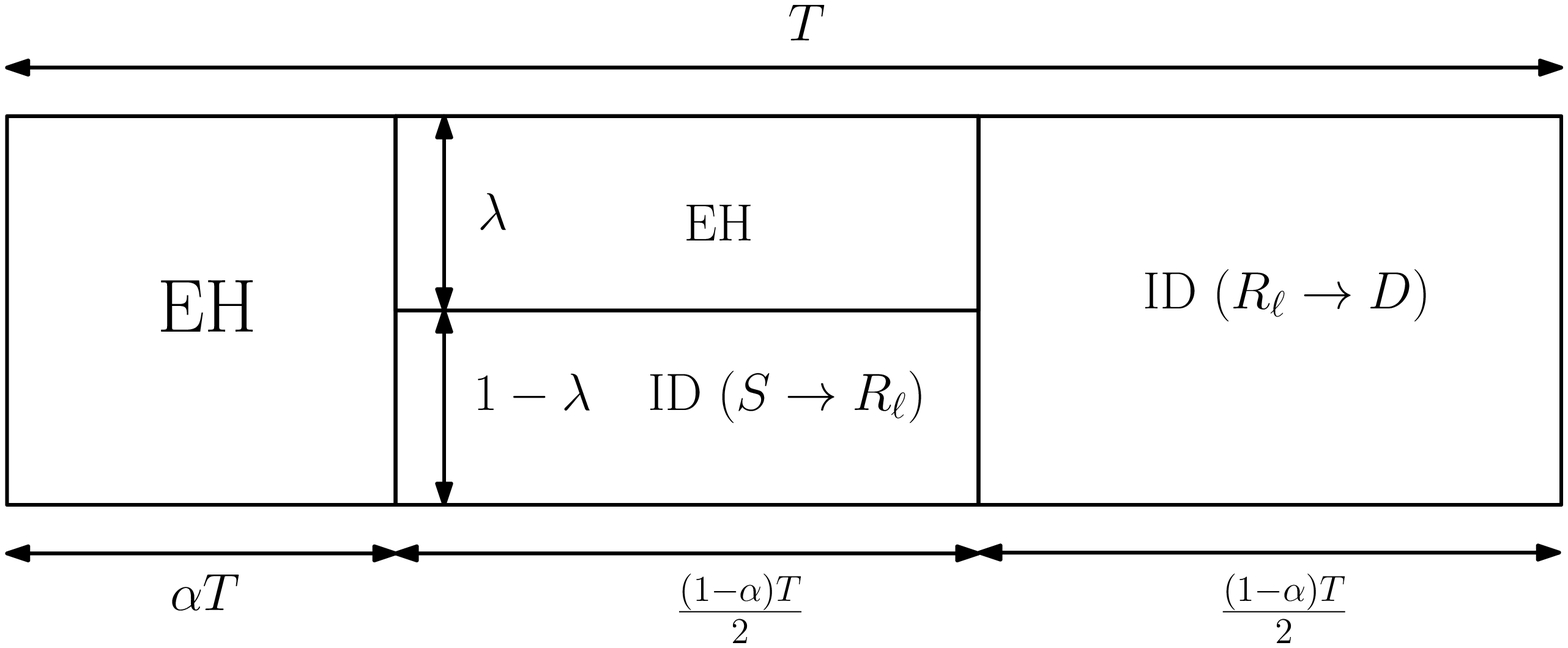}
		\caption{The frame structure of one slot of hybrid
EH protocol.}
		\label{time_frame}		
	\end{minipage}
 \end{figure}

\par As summarised in Fig \ref{time_frame}, data transmission from the source to the destination happens over three phases over a time slot of length $T$. In the first phase, over a duration of $\alpha T$, the source transmits data to all the $L$ relay nodes and the relays harvest this energy. Here, $\alpha$ is the time splitting (TS) factor. The signal received at the $\ell^{th}$ relay during $\alpha T$ is given by
\begin{equation}
    y_{1,\ell} = \sqrt{P_s d_{1,\ell}^{-\zeta}} g_{1,\ell} s + w_{a,\ell}+w_{c,\ell},
\end{equation}
where $P_s$ is the transmit power of the source node, $s$ is the signal transmitted, $\zeta$ is the path loss exponent, $w_{a,\ell}$ is the thermal noise and $w_{c,\ell}$ is the radio frequency (RF) to direct current (DC) conversion noise at the $\ell^{th}$ relay. Over the second phase of duration $\frac{(1-\alpha)T}{2}$, the source continues transmission to the relays. $\lambda$ (power splitting (PS) factor) fraction of the energy received over this phase is harvested and the rest is used for information decoding (ID). Thus the signal used for ID given by
\begin{equation}
    \Tilde{y}_{1,\ell} = \sqrt{1-\lambda} \left( \sqrt{P_s d_{1,\ell}^{-\zeta}} g_{1,\ell} s+ w_{a,\ell} \right)+w_{c,\ell} .
\end{equation}

For analytical tractability, we assume a linear relationship between the received and harvested energy. Such a linear model was considered in several works like \cite{liu2019cooperative,kumar2019performance,hadzi2015rate,medepally2010voluntary,ghosh2019outage,hadzi2015rate}. Thus, the total energy harvested at the $\ell^{th}$ relay over the two phases is given by,
\begin{equation}
    E_\ell = \eta \left(P_s d_{1,\ell}^{-\zeta} |g_{1,\ell}|^2 \right) T \left(\alpha + \lambda \frac{1-\alpha}{2} \right),
\end{equation} where $\eta$ is the efficiency of the EH circuit. In the third phase, a single relay is selected for information transmission to the destination. The relay which maximises the e2e SNR is chosen for transmission in this second hop (opportunistic relay selection) and the relay spends all the energy harvested from the previous phase to send the decoded information to the destination. Similar to \cite{atapattu2016optimal,liu2019cooperative,kumar2019performance}, we assume that the processing power required at the relays is negligible when compared to the power used for signal transmission. Thus, the transmit power available at the $\ell^{th}$ relay is given by 
\begin{equation}
    P_\ell = \frac{2 E_\ell}{(1-\alpha)T}.
\end{equation}
Then, the signal received at the destination from the $\ell^{th}$ relay can be written as
\begin{equation}
    y_{2,\ell} = \sqrt{P_\ell d_{2,\ell}^{-\zeta}} g_{2,\ell} s + w_D, 
\end{equation}where $w_D$ is the total additive noise at the destination (sum of thermal noise and the RF-to-
DC conversion noise). The SNR over the first and second hops of the $\ell^{th}$ relay is thus given by,
\begin{equation}
    \gamma_{1,\ell} = \frac{(1-\lambda)P_s d_{1,\ell}^{-\zeta}|g_{1,\ell}|^2}{\sigma_\ell^2} \ \text{and}
\end{equation}
\begin{equation}
    \gamma_{2,\ell} = \frac{P_\ell d_{2,\ell}^{-\zeta}|g_{2,\ell}|^2}{\sigma_D^2},
\end{equation}respectively. Here, $\sigma_\ell^2=(1-\lambda)\sigma_{a,\ell}^2+\sigma_{\ell,c}^2$ and $\sigma_D^2$ are the noise powers at the $\ell^{th}$ relay and destination respectively. The e2e SNR of the DF network when relay $\ell$ is transmitting in the second hop is defined as
\begin{equation}
    \gamma_{e2e,\ell} = min( \gamma_{1,\ell},\gamma_{2,\ell}).
    \end{equation}
With opportunistic relay selection, the index of the relay transmitting in the second hop can be written as
\begin{equation}
    \hat{\ell} = \underset{\ell=1,\cdots,L}{\text{argmax}} \  \gamma_{e2e,\ell},
\end{equation}
and 
\begin{equation}
\gamma_{e2e,max}^L = \max \{\gamma_{e2e,\ell};\ell=1,\cdots,L \}
\end{equation}
is the corresponding e2e SNR. Note that owing to the path loss component $d_{i,\ell}^{-\zeta}; \ i\in \{1,2\}, \ \ell\in\{1,\cdots,L\}$, the sequence of RVs $\{\gamma_{e2e,\ell} \ \ell\in\{1,\cdots,L\}\}$ are all independent but not identically distributed. In the next section, we derive the asymptotic distribution of $\gamma_{e2e,max}^L$, i.e the distribution of $\lim\limits_{L \to \infty}\gamma_{e2e,max}^L = \gamma_{e2e,max}$. Table \ref{symbols} summarises the frequently used symbols from the above model.

\begin{table}[b]
\begin{tabular}{|l|l|l|l|}
\hline
$S$ & Source node & $\lambda$            & Power Splitting (PS) factor \\ \hline
$D$          & Destination Node                                                                                    & $P_s$                & Transmit power at the source                                                                                                     \\ \hline
$L$          & Number of relays between $S$and $D$                                                                 & $P_\ell$             & Transmit power available at $R_\ell$                                                                                             \\ \hline
$R_\ell$     & $\ell$-th relay between $S$ and $D$                                                                  & $\eta$               & Efficiency of EH circuit at each relay                                                                                           \\ \hline
$g_{1,\ell}$ & Small scale fading gain of the link from $S$ to $R_\ell$ & $\sigma_\ell^2$      & Total noise power at $R_\ell$                                                                                                    \\ \hline
$g_{2,\ell}$ & Small scale fading gain of the link from $R_\ell$ to $D$ & $\sigma_D^2$         & Total noise power at $D$                                                                                                         \\ \hline
$d_{1,\ell}$ & Distance from $S$ to $R_\ell$                                                                       & $\gamma_{1,\ell}$    & SNR over the $S$ to $R_\ell$ link                                                                                                \\ \hline
$d_{2,\ell}$ & Distance from $R_\ell$ to $D$                                                                       & $\gamma_{2,\ell}$    & SNR over the $R_\ell$ to $D$ link                                                                                                \\ \hline
$\zeta$     & Path loss exponent                                                                                  & $\gamma_{e2e,\ell}$  & \begin{tabular}[c]{@{}l@{}}e2e SNR of the DF network\\ (=minimum of $\gamma_{1,\ell}$ and $\gamma_{2,\ell}$)\end{tabular} \\ \hline
$T$          & Length of a time slot                                                                               & $\gamma_{e2e,max}^L$ & Maximum e2e SNR over the $L$ relay links                                                                                  \\ \hline
$\alpha$     & Time Splitting (TS) factor                                                                          & $\gamma_{e2e,max}$   & \begin{tabular}[c]{@{}l@{}}Asymptotic maximum e2e SNR\\ (=$\lim\limits_{L \to \infty}\gamma_{e2e,max}^L$)\end{tabular}    \\ \hline
\end{tabular}
\caption{Table of frequently used symbols}
\label{symbols}
\end{table}

\section{Asymptotic Distribution of end-to-end SNR} \label{asymp_derive_sec}
 
{In this section, we first derive the asymptotic distribution of the normalised maximum e2e SNR, i.e we derive the distribution of $\Tilde{\gamma}_{e2e,max} := \lim\limits_{L \to \infty} \Tilde{\gamma}_{e2e,max}^L = \lim\limits_{L \to \infty} \frac{\gamma_{e2e,max}^L-b_L}{a_L} $ for normalising constants $a_L$ and $b_L$. Furthermore, using this result we study the statistical characterisation of the maximum e2e SNR RV $\gamma_{e2e,max}$.} Note that, here we need to evaluate the distribution of the maximum over a sequence of i.n.i.d. RVs $\{\gamma_{e2e,\ell}\}_{\ell=1}^L$. D.G.Mejzler in \cite{mejzler1969some} studied the asymptotic distribution of the maximal term of a sequence of RVs when each of the elements of the sequence are i.n.i.d.. Barakat et al. in \cite{barakat2013limit} extended this result to the case of random maximum of i.n.i.d. random vectors. The corresponding result is reproduced here for the case of maximum RV evaluated over $L$ univariate RVs. We first begin with the necessary uniformity assumptions (UAs) and then present the key result we utilise in Theorem.\ref{thm_barakat1}.\\
Let $Z_{max}^L=\text{max}\{Z_1,Z_2,\cdots,Z_L\}$ where $Z_\ell \sim F_\ell(z)$ for $\ell=1,\cdots,L$, then the distribution function (df) of $Z_{max}^L$ can be explicitly written as
\begin{equation}
    H_{max}^L(z) = \mathbb{P} \left( Z_{max}^L \leq z\right) = \prod\limits_{\ell=1}^L F_\ell(z) .
\end{equation}
The sequence $\{F_\ell({z})\}$ of dfs and the sequences ${a}_L \geq {0}$ and ${b}_L$ of normalising constants are said to satisfy the UAs for maximum vector ${Z}_{max}^L$ if 
\begin{equation}
\tag{$\mathcal{C}_1$}
    \underset{1 \leq \ell \leq L}{\text{max}}\{1-F_\ell(a_L z+b_L) \} \to 0 \ \text{as} \ L \to \infty,
    \label{ua1}
\end{equation} for all $a_Lz+b_L > \alpha(F_\ell)$ and $\alpha(F_\ell):= \text{inf}\{z:F_\ell(z)>0 \}>-\infty$. Also, for any fixed number $0<t\leq 1$ and each sequence of integers $\{m_L \}_L$ such that $m_L < L$, $m_L \to \infty$ and $\frac{m_L}{L} \to t$ as $L \to \infty$, we have that 
\begin{equation}
\tag{$\mathcal{C}_2$}
    \Tilde{u}(t,z) = \lim\limits_{L \to \infty} \sum\limits_{\ell=1}^{m_L} \left(1-F_\ell(a_Lz+b_L)) \right)
    \label{ua2}
\end{equation} exists and is finite for all $0 \leq t \leq 1$, whenever it is finite for $t=1$. 
These UAs are restrictions on the individual distribution functions $F_{\gamma_{e2e,1}}(\gamma), F_{\gamma_{e2e,2}}(\gamma),\cdots,F_{\gamma_{e2e,L}}(\gamma)$ as well as on the normalising constants
$a_L$ and $b_L$ necessary for non-trivial limit theorems \cite{barakat2013limit}. For instance, $\mathcal{C}_1$ is equivalent to  
\begin{equation}
\underset{1 \leq \ell \leq L}{\max} \mathbb{P} \left(\frac{Z_\ell - b_L}{a_L}  > z\right) \to 0 \ \text{as} \ L \to \infty.
\end{equation} 
Note that $\max\{Z_1,Z_2,\cdots,Z_L\} \to \max\{\omega(F_\ell),1 \leq \ell \leq L\}$ as $L \to \infty$ where $\omega(F_\ell)=\sup \{z: F_\ell(z)<1\}<\infty$. Hence, we need to choose appropriate normalising constants so that the RV $\frac{Z_{max}^L-b_L}{a_L}$ has a non-degenerate distribution when $L \to \infty$. The interested reader is encouraged to go through the proof of the above theorem in \cite{barakat2013limit} to better understand the role of $\mathcal{C}_1$ and $\mathcal{C}_2$ here.
\newline
Under the UA \ref{ua1} and \ref{ua2}, we have the following theorem \cite{barakat2013limit}: 
\begin{theorem} \label{thm_barakat1}
Under the UA \ref{ua1} and \ref{ua2}, a non-degenerate df $\tilde{H}_{max}(z)$ is the asymptotic distribution of $\frac{Z_{max}^L-b_L}{a_L}$ \ i.e 
\begin{equation}
    H_{max}^L\left(a_Lz+b_L \right) = \prod\limits_{\ell=1}^L F_\ell\left(a_Lz+b_L \right) \xrightarrow[]{\text{D}} \Tilde{H}_{max}(z) \ as \ L \to \infty,
\end{equation}where $\xrightarrow[]{\text{D}}$ stands for convergence in distribution if and only if 
\begin{equation}
    \Tilde{u}(z)=\Tilde{u}(1,z)= \lim\limits_{L \to \infty} \sum\limits_{\ell=1}^L \left(1-F_\ell\left(a_Lz+b_L \right) \right) < \infty.
    \label{cdf_form1}
\end{equation}Moreover, $\Tilde{H}_{max}(z)$ should have the form $\Tilde{H}_{max}(z)=e^{-\Tilde{u}(z)} $ and either $(i)$ $\log \Tilde{H}_{max}(z)$ is concave or $(ii)$ $\omega_{max}=\omega\left(\Tilde{H}_{max}(z) \right)$ is finite and $\log \Tilde{H}_{max}\left( \omega_{max}-e^{-z}\right)$ is concave or $(iii)$ $\alpha_{max}=\alpha\left(\Tilde{H}_{max}(z) \right)$ is finite and $\log \Tilde{H}_{max} \left( \alpha_{max}-e^{z}\right)$ is concave where $z>0$ in $(ii)$ and $(iii)$. 
\end{theorem}
\begin{proof}
Please refer \cite{barakat2013limit} for the proof.
\end{proof}

From the above theorem, it is clear that if we can {derive normalising constants $a_L$ and $b_L$ satisfying $\mathcal{C}_1$, $\mathcal{C}_2$ and (\ref{cdf_form1}) for $F_\ell(\gamma)=F_{\gamma_{e2e,\ell}}(\gamma)$, then we can arrive upon the asymptotic distribution of the RV $\Tilde{\gamma}_{e2e,max}^L$. }{Hence, we will use Theorem \ref{thm_barakat1} to determine the functional form of the distribution of $\Tilde{\gamma}_{e2e,max}$ and then use that distribution to arrive upon the distribution of $\gamma_{e2e,max}$.}
\newline
We begin by characterising the distribution of $F_{\gamma_{e2e,\ell}}(\gamma)$.
According to the system model discussed in Section II, the channel fading coefficients are Rayleigh distributed and hence $|g_{1,\ell}|^2$ will be exponentially distributed with unit mean \cite[Chapter 5]{molisch2012wireless}. 

Then, by making use of the scaling properties of the exponential random variables, we have $\gamma_{1,\ell}  \sim \text{Exp}(\theta_\ell)$, where $\text{Exp}(\theta)$ represents the exponential distribution with scale parameter $\theta$ and 
\begin{equation}
     {\theta_\ell}  = \frac{\sigma_\ell^2}{(1-\lambda)P_s d_{1,\ell}^{-\zeta}}.
\end{equation}
Similarly, note that $  \gamma_{2,\ell}=  \gamma_{1,\ell} \times \varphi_{2,\ell}$, where \begin{equation}
    \varphi_{2,\ell}=\frac{\eta\sigma_{\ell}^2}{1-\lambda} \left(\lambda+\frac{2\alpha}{1-\alpha} \right)  \frac{ d_{2,\ell}^{-\zeta}|g_{2,\ell}|^2}{\sigma_D^2}.
\end{equation}
Hence, $\gamma_{2,\ell}$ is the product of two exponential random variables with scale parameters $\theta_\ell$ and $\nu_\ell$ where 
\begin{equation}
     {\nu_\ell}  = \frac{1-\lambda}{\eta\sigma_{\ell}^2} \left(\lambda+\frac{2\alpha}{1-\alpha} \right)^{-1}  \frac{\sigma_D^2}{ d_{2,\ell}^{-\zeta}}.
\end{equation}
Thus, we have the following lemma giving the distribution of $\gamma_{e2e,\ell}$ :
\begin{lemma} \label{cdf_e2e_thm}
The CDF of the RV $\gamma_{e2e,\ell} = \min\{\gamma_{1,\ell},\gamma_{1,\ell} \times \varphi_{2,\ell}\}$ where $\gamma_{1,\ell} \sim \text{Exp}(\theta_\ell)$ and $\varphi_{2,\ell} \sim \text{Exp}(\nu_\ell)$ is given by 
\begin{equation}
      F_{\gamma_{e2e,\ell}}(\gamma)  =1-\theta_\ell \sum\limits_{k=0}^ \infty \frac{(-\nu_\ell)^k}{k!} \gamma E_k(\theta_\ell \gamma);  \ \gamma \geq 0
    \label{cdf_e2e1}
\end{equation} where $E_k(.)$ is the exponential integral function.
\end{lemma}
\begin{proof}
Please refer Appendix A for the proof. 
\end{proof}
Using the above CDF, we derive the distribution of $\Tilde{\gamma}_{e2e,max}$ and the results are presented in the following theorem :

\begin{theorem} \label{asymp_thm}
The distribution of $\tilde{\gamma}_{e2e,max}$ is given by 
\begin{align}
     & F_{\Tilde{\gamma}_{e2e,max}}(\gamma)  =\exp(-\exp(-\gamma)),
    \label{uz_final}
\end{align} for the choice of normalising constants $a_L={\beta}^{-1}$ and $b_L={\beta}^{-1} \log\left(\sum\limits_{\ell=1}^L \mathbb{I}_{\theta_\ell=\beta} \exp(-\nu_\ell) \right)$. Here, $\beta=\theta_{\hat{\ell}} \in \{\theta_1,\cdots, \theta_L\}$ is the smallest possible value of $\beta$ such that 
\begin{equation}
\lim\limits_{L \to \infty} \sum\limits_{\ell=1}^L \mathbb{I}_{\theta_\ell=\beta} \to \infty \ \text{as} \ L \to \infty,
\quad \text{and} \quad  \lim\limits_{L \to \infty} R_\ell \left(R_{\hat{\ell}}\right)^{\frac{-\theta_\ell}{\beta}} = 0 \ \forall \  \theta_\ell > \beta
\end{equation}
are satisfied. Here, $R_\ell = \sum\limits_{i=1}^L \mathbb{I}_{\theta_i=\theta_\ell}$ and $\mathbb{I}_{A}$ represents the indicator function of the event $A$.
\end{theorem}
\begin{proof}

Recall from Theorem \ref{thm_barakat1} that if we can identify normalising constants $a_L\geq 0$ and $b_L$ satisfying $\mathcal{C}_1$ and $\mathcal{C}_2$ such that
\begin{equation}
    \Tilde{u}(\gamma)=\lim\limits_{L \to \infty} \sum\limits_{\ell=1}^L \left(1-F_\ell\left(a_L \gamma +b_L \right) \right) < \infty,
    \label{uz_appendix}
\end{equation}
then we can identify the form of the asymptotic distribution of $\Tilde{\gamma}_{e2e,max}^L$. Mejzlers theorem \cite[Chapter 5]{de2007extreme} gives specific conditions on the normalising constants $a_L$ and $b_L$ such that the UA's (\ref{ua1}) and (\ref{ua2}) are satisfied. Using these results, we assume that there exist sequences $a_L$ and $b_L$ such that
\begin{align}
\mid \log a_L \mid + \mid b_L \mid \ \to \infty \ \text{as} \ L \to \infty
\label{norm_contn1}
\end{align}
and
\begin{equation}
\begin{array}{l}
\frac{a_{L+1}}{a_{L}} \rightarrow 1, \\
\frac{\left(b_{L+1}-b_{L}\right)}{a_{L}} \rightarrow 0
\end{array}
\label{norm_contn2}
\end{equation}are satisfied. Note that the above conditions ensure that the normalising constants grows as a smooth function of $L$. They also guarantee that the normalising constants are large enough so that $\tilde{\gamma}_{e2e,max}^L$ has a non degenerate distribution as $L \to \infty$. Next, for this choice of $a_L$ and $b_L$, we evaluate $\Tilde{u}(\gamma)$. Hence we have,
\begin{align}
    \Tilde{u}(\gamma) =\lim\limits_{L \to \infty}  \sum\limits_{\ell=1}^L  \theta_\ell \sum\limits_{k=0}^\infty \frac{(-\nu_\ell)^k}{k!} (a_L \gamma + b_L) \ E_{k}(\theta_\ell(a_L\gamma+ b_L)).
    \label{ugamma_expand}
\end{align}
Since $a_L$ and $b_L$ satisfy the conditions in (\ref{norm_contn1}) and (\ref{norm_contn2}), the argument of the exponential integral function increases to infinity as $L \to \infty$. Hence, we make use of the following asymptotic expansion of the exponential integral function to expand  (\ref{ugamma_expand}):
\begin{equation}
    E_n(x) = \frac{\exp(-x)}{x}\left \lbrace 1-\frac{n}{x} +\frac{n(n+1)}{x^2} - \cdots \right  \rbrace.
    \label{exp_int_n}
\end{equation}

Using the above expansion, we rewrite $\Tilde{u}(\gamma)$ as 
\begin{align}
 \Tilde{u}(\gamma) &= \lim\limits_{L \to \infty} \sum\limits_{\ell=1}^L \   \sum\limits_{k=0}^\infty \frac{(-\nu_\ell)^k}{k!} \exp(-\theta_\ell (a_L \gamma+b_L)) \ \left \lbrace 1-\frac{k}{\theta_\ell (a_L\gamma+b_L)}  + \frac{k(k+1)}{\left(\theta_\ell (a_L \gamma+b_L) \right)^2} -\cdots \right \rbrace, 
     \label{uz_expand1}\\ 
     & = \lim\limits_{L \to \infty} \sum\limits_{\ell=1}^L \   \sum\limits_{k=0}^\infty \frac{(-\nu_\ell)^k}{k!} \exp(-\theta_\ell (a_L \gamma+b_L)) \ \left \lbrace 1+ \mathcal{O} \left( \frac{1}{a_L \gamma + b_L}\right)  \right \rbrace.
\end{align}
We know that $a_L \gamma + b_L \to \infty$ as $L \to \infty$ and hence we approximate $\tilde{u}(\gamma)$ by approximating all terms of the form $\left( \frac{1}{a_L\gamma+b_L}\right)^k$ for all $k>0$ to be zero.
Thus we have,
\begin{align}
\Tilde{u}(\gamma) & = \lim\limits_{L  \to \infty} \sum\limits_{\ell=1}^L \exp(-\theta_\ell (a_L \gamma+b_L)) \sum\limits_{k=0}^\infty \frac{(-\nu_\ell)^k}{k!} \\ &= \lim\limits_{L  \to \infty} \sum\limits_{\ell=1}^L \exp(-\theta_\ell (a_L \gamma+b_L)) \exp(-\nu_\ell).
\label{u_gamma_final_limit}
\end{align}
Thus, the asymptotic distribution of the RV $\Tilde{\gamma}_{e2e,max}^L$ can be determined by identifying $a_L$ and $b_L$ satisfying (\ref{norm_contn1}) and (\ref{norm_contn2}) such that the limit in (\ref{u_gamma_final_limit}) evaluates to a function of $\gamma$ taking finite values. Hence in the following paragraphs, we identify possible choices of $a_L$ and $b_L$ for various special cases which are of practical interest.

\subsubsection{All links are i.i.d.} \label{iid}
We begin with the most simple case where $\{\theta_\ell\}$ and $\{\nu_\ell\}$ are identical. Hence, we have $\theta_\ell=\theta$ and $\nu_\ell=\nu$ $\forall \ \ell \in \{1,\cdots,L\}$. In this case, the limit in (\ref{u_gamma_final_limit}) is given by
\begin{equation}
\Tilde{u}(\gamma) = \lim\limits_{L \to \infty} L \exp(-\theta (a_L \gamma+b_L)) \exp(-\nu).
\end{equation}
Here, if we choose $a_L=\frac{1}{\theta}$ and $b_L=\frac{\log(L)-\nu}{\theta}$, we have
\begin{equation}
\Tilde{u}(\gamma) = \lim\limits_{L \to \infty} L \exp(- \gamma) \exp(-\log(L)+\nu) \exp(-\nu) = \exp(- \gamma).
\end{equation}

{Note that this choice of $a_L$ and $b_L$ satisfy the conditions in (\ref{norm_contn1}) and (\ref{norm_contn2}).} Thus, the distribution of the RV $\Tilde{\gamma}_{e2e,max}$ in this case is given by 
\begin{equation}
F_{\Tilde{\gamma}_{e2e,max}}(\gamma)=\exp(-\exp(- \gamma)),
\end{equation} when $a_L=\frac{1}{\theta}$ and $b_L = \frac{\log(L)-\nu}{\theta}$.
Fig (\ref{norm_iid_L_60}) shows one simulation using this result.
\begin{figure}[t]
    \centering
    \includegraphics[scale=0.5]{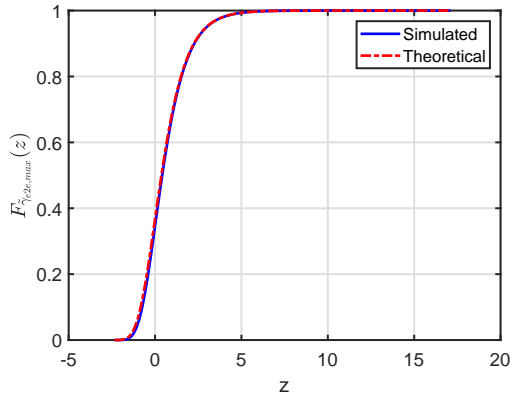}
    \caption{CDF of $\tilde{\gamma}_{e2e,amx}$ for $L=60$, $R_{\theta}=2$ and $R_{\hat{\ell}}=20$.}
    \label{norm_iid_L_60}
\end{figure}
\subsubsection{All S to $R_\ell$ links are i.i.d.} \label{nu_inid}
According to our system model, such a scenario can arise in cases where all relays are deployed at equal distances from the source node. One such model would be the system where all the relays are placed along the periphery of a circle of radius $d_{1,\ell}=d_1 \forall \ \ell$ with the source $S$ at the centre of the circle. In this case we have $\theta_\ell=\theta$ $\forall \ell$. Following the initial steps similar to the previous case, we have
\begin{align}
\Tilde{u}(\gamma) & = \lim\limits_{L \to \infty} \sum\limits_{\ell=1}^L \exp(-\theta (a_L \gamma+b_L)) \sum\limits_{k=0}^\infty \frac{(-\nu_\ell)^k}{k!} \\
& = \lim\limits_{L \to \infty} \exp(-\theta(a_L \gamma+b_L)) \sum\limits_{\ell=1}^L \exp(-\nu_\ell).
\end{align}
Now, if we choose $a_L=\frac{1}{\theta}$ and $b_L=\frac{\log \left({\sum_{\ell=1}^L \exp(-\nu_\ell)}\right)}{\theta}$, {which satisfy the conditions in (\ref{norm_contn1}) and (\ref{norm_contn2}),} we have
\begin{equation}
\Tilde{u}(\gamma) = \exp(- \gamma). 
\end{equation}

Again, the distribution of the RV $\Tilde{\gamma}_{e2e,max}$ in this case is given by 
\begin{equation}
F_{\Tilde{\gamma}_{e2e,max}}(\gamma)=\exp(-\exp(- \gamma)),
\end{equation} 
which is log concave.
Fig (\ref{norm_iid_nu_L_60}) shows one simulation using this result.
\begin{figure}[t]
    \centering
    \includegraphics[scale=0.5]{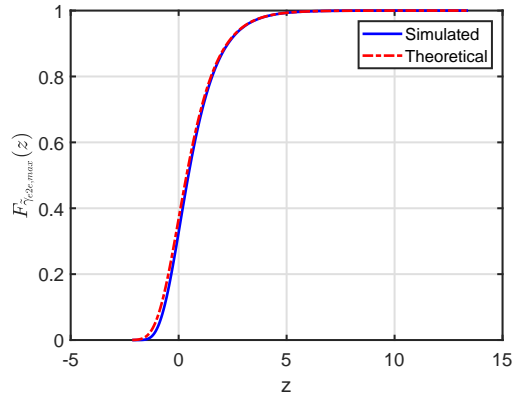}
    \caption{CDF of $\tilde{\gamma}_{e2e,amx}$ for $L=60$, $R_{\theta}=2$ and $R_{\hat{\ell}}=20$.}
    \label{norm_iid_nu_L_60}
\end{figure}
\subsubsection{All $R_\ell$ to D links are i.i.d.} \label{theta_inid}
In this case, we have $\nu_\ell=\nu$ $\forall \ \ell$ and hence we have
\begin{equation}
\Tilde{u}(\gamma) = \lim\limits_{L \to \infty} \exp(-\nu) \sum\limits_{\ell=1}^L \exp(-\theta_\ell(a_L \gamma + b_L)).
\end{equation}
Now, in any practical system we can assume that the set of possible values of $\theta_\ell$ represented by $\Theta$ has finite cardinality $R_{\Theta}$. For example according to the system model considered in Section II, $\theta_\ell$ are i.n.i.d. because of the difference in distances between the source and the $\ell$-th relay. Now, we can uniformly quantize the distance between the relay closest and farthest from $S$ such that $d_{1,\ell}^{-\zeta}$ and hence $\theta_\ell$ are not very different for relays falling into the same quantization bin. Recall that $\zeta$ is the path loss coefficient and $\theta_\ell \propto d_{1,\ell}^\zeta$. Let $R_\ell$ be defined as follows:
\begin{equation}
 R_\ell = \sum\limits_{i=1}^L \mathbb{I}_{\theta_i=\theta_\ell},
\end{equation}
where $\mathbb{I}_{X}$ is the indicator function of the event $X$. Thus, $R_\ell$ gives the number of occurences of the value $\theta_\ell$ among $\{\theta_i\}_{i=1}^L$. Note that  in any practical system $R_\ell$ will be a non decreasing function of $L$ and the maximum rate at which $R_\ell$ can grow will be $\mathcal{O}(L)$. Next, we choose $\theta_{\hat{\ell}}=\beta$ to be the smallest element in ${\Theta}$ that satisfies the following two conditions:
\begin{equation}
\lim\limits_{L \to \infty} \sum\limits_{\ell=1}^L \mathbb{I}_{\theta_\ell=\beta} \to \infty \ \text{as} \ L \to \infty,
\label{beta_1}
\end{equation}
and
\begin{equation}
 \lim\limits_{L \to \infty} R_\ell \left(R_{\hat{\ell}}\right)^{\frac{-\theta_\ell}{\beta}} = 0 \ \forall \theta_\ell > \beta.
\label{beta2}
\end{equation}
Now, when $L \to \infty$ and with finite choices of $\theta_\ell$, we will have at least one $\beta$ that satisfies the above two conditions in any practical system \footnote{Note that if any one of the following conditions is true, we need to consider only  (\ref{beta_1}) for choosing $\beta$: (i) if $R_{\hat{\ell}}=\mathcal{O}(L)$ when $\theta_{\hat{\ell}}$ is the smallest $\theta_\ell$ satisfying (\ref{beta_1}) or (ii) if $R_{\hat{\ell}}=\mathcal{O}(f(L))$ for all $\ell$; i.e $R_\ell$ corresponding to all $\theta_\ell$ satisfying (\ref{beta_1}) has the same rate of growth as $M$ increases. }. Next we choose $a_L=\frac{1}{\beta}$  $b_L=\frac{\log(R_{\hat{\ell}})-\nu}{\beta}$. {Furthermore, note that this choice of $a_L$ and $b_L$ satisfy the conditions in (\ref{norm_contn1}) and (\ref{norm_contn2}).} Thus we have
\begin{align}
\Tilde{u}(\gamma) & = \exp(-\nu) \lim\limits_{L \to \infty} \sum\limits_{\ell=1}^L \exp(-\theta_\ell a_L \gamma) \exp\left(-\theta_\ell \frac{\log(R_{\hat{\ell}})}{\beta} \right) \exp\left(\theta_\ell \frac{\nu}{\beta} \right)\\
&= \exp(-\nu) \lim\limits_{L \to \infty} \underbrace{\sum\limits_{\ell=1}^L \exp\left(-\theta_\ell \frac{\gamma}{\beta} \right) R_{\hat{\ell}}^{\frac{-\theta_\ell}{\beta}}\exp\left(\theta_\ell \frac{\nu}{\beta} \right)}_{Term 1}. 
\label{uz_theta_inid}
\end{align}
Furthermore, we rewrite  (\ref{uz_theta_inid}) as follows:
\begin{equation}
\Tilde{u}(\gamma) = \exp(-\nu) \lim\limits_{L \to \infty} \left \lbrace \sum\limits_{i=1}^{R_{\Theta}} R_i \exp\left(-\theta_i \frac{\gamma}{\beta}\right) R_{\hat{\ell}}^{\frac{{-\theta}_i}{\beta}}\exp\left(\theta_i \frac{\nu}{\beta} \right) \right \rbrace  .
\end{equation}
Now, we have  $\lim\limits_{L \to \infty} R_i R_{\hat{\ell}}^{\frac{-\theta_i}{\beta}} =0$ for all $i > \hat{\ell}$. Similarly, for all $i< \hat{\ell}$, $R_i$ is smaller than $R_{\hat{\ell}}$ so that $\lim\limits_{L \to \infty} R_i R_{\hat{\ell}}^{\frac{-\theta_i}{\beta}} =0$. Thus,
\begin{equation}
\Tilde{u}(\gamma) =  \exp(- \gamma).
\end{equation} Thus, the distribution of $\Tilde{\gamma}_{e2e,max}$ in this case is given by, 
\begin{equation}
F_{\Tilde{\gamma}_{e2e,max}}(\gamma)=\exp(-\exp(- \gamma)),
\end{equation} when $a_L=\frac{1}{\beta}$ and $b_L=\frac{\log(R_{\hat{\ell}})-\nu}{\beta}$.
Fig (\ref{norm_iid_theta_L_60}) shows one simulation using this result.
\begin{figure}[t]
    \centering
    \includegraphics[scale=0.5]{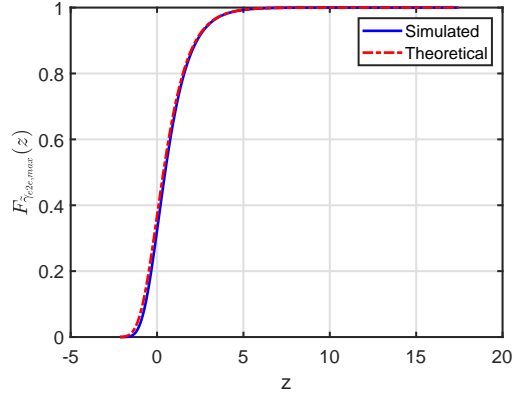}
    \caption{CDF of $\tilde{\gamma}_{e2e,amx}$ for $L=60$.}
    \label{norm_iid_theta_L_60}
\end{figure}
\subsubsection{All links are i.n.i.d.} \label{inid}
In this case, we have 
\begin{equation}
\Tilde{u}(\gamma) = \lim\limits_{L \to \infty} \sum\limits_{\ell=1}^L \exp(-\theta_\ell a_L \gamma)\exp(-\theta_\ell b_L) \exp(-\nu_\ell).
\end{equation}
As in the previous case we assume that there are only $R_{\Theta}$ possible values for $\theta_\ell$ and there exists $\beta$ such that (\ref{beta_1}) and (\ref{beta2}) are satisfied. Choosing $a_L=\frac{1}{\beta}$, we have 
\begin{equation}
\Tilde{u}(\gamma) = \lim\limits_{L \to \infty} \sum\limits_{\ell=1}^L \exp\left(-\frac{-\theta_\ell \gamma}{\beta}\right) \exp(-\nu_\ell) \exp(-\theta_\ell b_L)
\end{equation}
Rewriting the sum term in the previous expressions in terms of the distinct $\theta_\ell$'s 
\begin{equation}
\Tilde{u}(\gamma) = \lim\limits_{L \to \infty} \sum\limits_{i=1}^{R_{\Theta}} \exp\left(\frac{\theta_i \gamma}{\beta} \right) \exp(-\theta_i b_L) \sum\limits_{k:\theta_k = \theta_i}\exp(-\nu_k).
\end{equation}
Next, we choose $b_L$ as follows:
\begin{equation}
b_L = \frac{1}{\beta} \log\left(\sum\limits_{j:\theta_j=\beta} \exp(-\nu_j) \right).
\end{equation}
{Again, we can see that this choice of $a_L$ and $b_L$ satisfy the conditions in (\ref{norm_contn1}) and (\ref{norm_contn2}).} Using the above choice of $b_L$, $\Tilde{u}(\gamma)$ can be rewritten as 
\begin{equation}
\Tilde{u}(\gamma) = \lim\limits_{L \to \infty} \sum\limits_{i=1}^{R_{\Theta}} \underbrace{\exp\left(\frac{-\theta_i \gamma}{\beta} \right) \left(\sum\limits_{j:\theta_j=\beta} \exp(-\nu_j)\right)^{-\frac{\theta_i}{\beta}} \sum\limits_{k:\theta_k=\theta_i} \exp(-\nu_k)}_{\text{Term 2}}
\end{equation}
Next, let us analyse Term 2 for different values of $\theta_i$. For $\theta_i=\beta$, term2 will reduce to $\exp(-\gamma)$. Now for $\theta_i \neq \beta$, we have
\begin{equation}
\exp\left(\frac{-\theta_i\gamma}{\beta}\right)  \frac{\sum\limits_{k:\theta_k=\theta_i} \exp(-\nu_k)}{\left(\sum\limits_{j:\theta_j=\beta} \exp(-\nu_j)\right)^{-\frac{\theta_i}{\beta}}} \leq \exp\left(\frac{-\theta_i \gamma}{\beta} \right) \underbrace{\frac{R_i \exp(-\nu_{min})}{(R_{\hat{\ell}}\exp(-\nu_{min}))^{-\theta_i/\beta}}}_{\text{Term 3}}
\end{equation}
We know that $\lim\limits_{L \to \infty} R_{\hat{\ell}}=\infty$. Now, whenever $\theta_i > \beta$, we can observe that the denominator of Term 3 will grow to infinity faster than the numerator and hence Term 3 will converge to zero.
Thus, we conclude that when $\theta_i \neq \beta$ and $\theta_i > \beta$, Term 2 will be zero. Now if $\theta_i < \beta$, then $R_i < R_{\hat{\ell}}$ and we know that $\lim\limits_{L \to \infty} R_i< \infty$. Hence, the denominator of Term 3 will be zero even for this case as $L \to \infty$. Thus we can conclude that 
\begin{equation}
\Tilde{u}(\gamma) =  \exp(- \gamma).
\label{uz_final_inid}
\end{equation}
The asymptotic distribution of the normalised maximum e2e SNR is then given by $F_{\Tilde{\gamma}_{e2e,max}}(\gamma)=\exp(-\exp(-\gamma))$ when $a_L=\frac{1}{\beta}$ and $b_L=\frac{1}{\beta} \log\left(\sum\limits_{j:\theta_j=\beta} \exp(-\nu_j) \right)$.
Fig (\ref{norm_inid_all_L_60}) shows one simulation using this result.
\begin{figure}[t]
    \centering
    \includegraphics[scale=0.5]{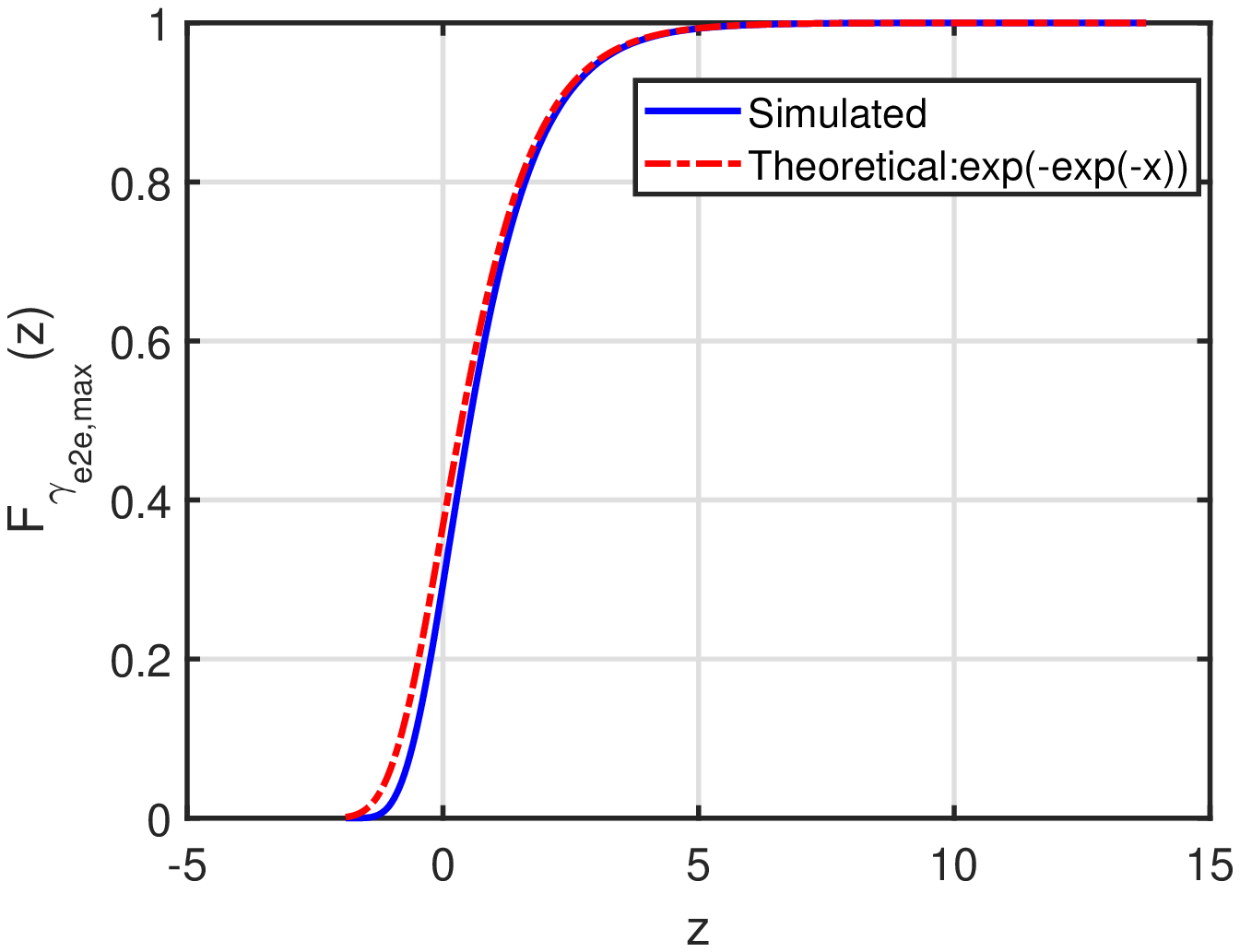}
    \caption{CDF of $\tilde{\gamma}_{e2e,amx}$ for $L=60$, $R_{\theta}=2$ and $R_{\hat{\ell}}=20$.}
    \label{norm_inid_all_L_60}
\end{figure}
\end{proof}

{Note that we need the statistics of the unnormalised RV $\gamma_{e2e,max}$ (which is the maximum e2e SNR) for all practical analysis and this CDF can be evaluated as $F_{\gamma_{e2e,max}} = \exp\left(-\exp\left(-\frac{(\gamma-b_L)}{a_L} \right) \right)$. This is same as evaluating $\exp\left(-\Tilde{u}\left(\frac{\gamma-b_L}{a_L} \right) \right)$. Thus, this CDF can also be computed as $F_{\gamma_{e2e,max}} = \exp(-u(\gamma))$ where
\begin{equation}
u(\gamma) = \sum\limits_{\ell=1}^L \exp(-\theta_\ell \gamma -\nu_\ell),
\label{uz_new}
\end{equation}
for moderate values of $L$. Note that the distribution proposed in all the cases discussed above are far easier to evaluate than the exact distribution of the maximum RV given by }
    \begin{equation}
        F_{\gamma_{e2e,max}^L}(\gamma) = \prod \limits_{\ell=1}^L 1-\theta_\ell \sum\limits_{k=0}^ \infty \frac{(-\nu_\ell)^k}{k!} \gamma E_k(\theta_\ell \gamma).
       \label{exact_max}
    \end{equation}

{Here, each of the product term involves an infinite summation and hence (\ref{exact_max}) will have a 
complicated expression which is difficult to evaluate whenever $L>2$. Therefore, a simplified expression for the CDF of ${\gamma_{e2e,max}}= \lim\limits_{L \to \infty} \gamma_{e2e,max}^L$ will prove to be propitious for analysing system performance and resource planning. While (\ref{uz_new}) corresponds to the more general case, for all the special cases mentioned in Section \ref{iid} - Section \ref{inid} we can use the Gumbel CDF for all the analysis. Furthermore from an engineering point of view, one can always plan where to place the relay and hence we expect to encounter the scenario in Section \ref{nu_inid} more commonly.}

\subsection{Asymptotic ergodic and outage capacity} \label{erg_cap_out_cap}

Given that we have characterised the distribution of the asymptotic e2e SNR, we proceed to derive the asymptotic ergodic capacity defined as 
\begin{equation}
    C_{e2e,max} =\frac{1}{2}  \lim\limits_{L \to \infty}  \mathbb{E}\left[\log_2\left(1+\gamma_{e2e,max}^L \right) \right].
    \label{asym_erg_temp}
\end{equation} Here, the factor $\frac{1}{2}$ accounts for the transmission of information happening only over half of the time slot \cite{gu2015rf}. Since we have already proved that $\gamma_{e2e,max}^L \underrightarrow{D} \gamma_{e2e,max}$, the asymptotic ergodic capacity can be evaluated as 
\begin{equation}
    C_{e2e,max} = \frac{1}{2} \mathbb{E}\left[\log_2\left(1+\gamma_{e2e,max} \right) \right].
    \label{asym_erg}
\end{equation}
This result can be easily derived by applying continuous mapping theorem and monotone convergence theorem to the expression in (\ref{asym_erg_temp}). The steps are very similar to the proof in Appendix B of \cite{subhash2020transmit} and hence we do not repeat them here. Now, the expression in (\ref{asym_erg}) can be evaluated via numerical integration of the following expression :
\begin{equation}
    C_{e2e,max} = \frac{1}{2} \times \int\limits_0^\infty \log_2\left(1+\gamma \right) \ f_{\gamma_{e2e,max}}(\gamma) \ d\gamma.
     \label{asym_erg_intg}
\end{equation}Here, $f_{\gamma_{e2e,max}}(\gamma)$ is the pdf of ${\gamma_{e2e,max}}$ and is given below :
\begin{equation}
    f_{\gamma_{e2e,max}}(\gamma) = \exp \left(-\sum\limits_{\ell=1}^Le^{-(\theta_\ell \gamma+\nu_\ell)} \right) \times  \sum\limits_{\ell=1}^L \theta_\ell e^{-(\theta_\ell \gamma+\nu_\ell)}.
\end{equation}

{Note that for all the special cases discussed in Section \ref{iid}-Section \ref{inid}, we can use the Gumbel pdf for $f_{\gamma_{e2e,max}}$.} The effective information transmission time decides the achievable throughput, which in this case is defined as follows \cite{gu2015rf}
\begin{equation}
    R_{e2e,max} = (1-\alpha) C_{e2e,max}.
    \label{re2emax}
\end{equation}
Similarly, we can characterise the asymptotic outage capacity of the system using the statistics of the maximum e2e SNR.  Outage capacity is defined as the maximum constant rate that can be maintained over the fading blocks with a specified outage probability \cite{gu2015rf}\footnote{Note that this is particularly useful for slowly varying channels, where the instantaneous SNR remains constant over a  large number of symbols.}. Here, the outage capacity is given by

\begin{equation}
C_{e2e,max}^{out}=\frac{1}{2}\left[1-P_{e2e,max}^{out}(\gamma_{th}) \right] \log_2\left(1+\gamma_{th} \right),
\end{equation}
where $P_{e2e,max}^{out}(\gamma_{th}) $ is the probability of outage for a threshold of $\gamma_{th}$; i.e  $P_{e2e,max}^{out}(\gamma_{th})=\mathbb{P}\left(\gamma_{e2e,max} \leq \gamma_{th} \right)$. Hence, the outage capacity can be easily derived from the asymptotic CDF of $\gamma_{e2e,max}$ as 

\begin{equation}
    C_{e2e,max}^{out} = \frac{\log_2\left(1+\gamma_{th} \right)}{2} \left( 1-F_{\gamma_{e2e,max}}(\gamma_{th}) \right).
\end{equation}

\subsection{Ordering of asymptotic e2e SNR} \label{ordering_e2esnr}

Stochastic ordering allows ordering of RVs with respect to the variations in their parameters. An RV $X$ is said to be stochastically smaller than an RV $Y$ if
\begin{equation}
    \mathbb P(X > z) \leq \mathbb P(Y >z), \ \forall z \in \mathbb{R},
\end{equation} and is written as $X \leq_{st} Y$ \cite{shaked2007stochastic} .
Such an ordering of SNR RVs allows us to study the variations of SNR and hence functions of SNR with changes in the different channel parameters. This will be highly useful for system planning and resource allocation without much computational burden every time a decision is to be made. Stochastic ordering has been effectively used for analysing the performance of various communication systems in works like \cite{srinivasan2018secrecy,srinivasan2019analysis,srinivasan2019analysismrc}.\\
In the following subsections, we establish the stochastic ordering of $\gamma_{e2e, max}$ with respect to variations in the following parameters: (1) source transmit power $P_s$, (2) noise variance $\sigma^2$, (3) TS factor $\alpha$ and (4) PS factor $\lambda$. For all further analysis we make the assumption that $\sigma_D^2 = \sigma_\ell^2=\sigma^2 \ \forall \ \ell \in \{1,\cdots,L \}.$

\subsubsection{Ordering with respect to $P_s$} \label{order1}
 
Let $X_1$ and $X_2$ be the RVs representing the asymptotic maximum e2e SNR with the transmit power $P_1$ and $P_2$ respectively. We further assume that $P_1>P_2$ and the rest of the parameters are considered to be the same for both the RVs\footnote{Note that the RVs $X_1$ and $X_2$ will have the same set of parameters $\{\nu_\ell;\ell=1,\cdots,L\}$ since they are independent of the source transmit power.}. Hence, we have $X_1$ with parameters $\{\theta_\ell^{(1)},\nu_\ell;\ell=1,\cdots,L\}$ and  $X_2$ with parameters  $\{\theta_\ell^{(2)},\nu_\ell;\ell=1,\cdots,L\}$ where $\theta_\ell^{(i)}=\frac{\sigma_\ell^2}{(1-\lambda)P_id_{1,\ell}^{-\zeta}}$ for $i \in \{1,2\}$. $X_2 \leq_{st} X_1$ if the following is true: 
\begin{align}
	\exp\left( -\sum\limits_{\ell=1}^L \exp{(-\theta_\ell^{(2)} z -\nu_\ell)}\right) & \geq 	\exp\left( -\sum\limits_{\ell=1}^L \exp{(-\theta_\ell^{(1)} z -\nu_\ell)}\right) \\ i.e \ 
	\sum\limits_{\ell=1}^L \exp\left(-\frac{\Tilde{\theta}_\ell z}{P_2} -\nu_\ell\right) & \leq \sum\limits_{\ell=1}^L \exp\left(-\frac{\Tilde{\theta}_\ell z}{P_1} -\nu_\ell\right),
	\label{order_ps_eqn2}
\end{align}
where 
\begin{equation}
{\theta}_\ell^{(i)}  = \frac{\Tilde{\theta}_\ell}{P_i}.
\end{equation}
Upon further rearrangement, (\ref{order_ps_eqn2}) can be re-written as follows
\begin{equation}
	\sum\limits_{\ell=1}^L \exp(-\nu_\ell)\left \lbrace \underbrace{\exp\left(-\frac{\Tilde{\theta}_\ell z}{P_2} \right) - \exp\left(-\frac{\Tilde{\theta}_\ell z}{P_1}\right)}_\text{Term 1}\right \rbrace \leq 0.
	\label{order_ps}
\end{equation}
Note that for $P_1>P_2$ and $z>0$ term 1 will be negative for all values of $\ell$ and hence the inequality in (\ref{order_ps}) holds for all choices of $\Tilde{\theta_\ell}$, $z$ and $\nu_\ell$. Thus, we conclude that $X_2$ is stochastically smaller than $X_1$ when $P_2<P_1$. Note that this observation is intuitive since the e2e SNR is expected to increase with the increase in transmit power. Nevertheless, this result reaffirms the utility of our asymptotic result in deriving meaningful inferences about the system performance with respect to different system parameters.  

\subsubsection{Ordering with respect to $\alpha$} \label{order2}

Let $X_1$ and $X_2$ be the RVs representing the asymptotic maximum e2e SNR with the TS factor $\alpha_1$ and $\alpha_2$ respectively. We further assume that $\alpha_1>\alpha_2$, RV $X_1$ has parameters $\{\theta_\ell,\nu_\ell^{(1)};\ell=1,\cdots,L\}$ and  $X_2$ has parameters  $\{\theta_\ell,\nu_\ell^{(2)};\ell=1,\cdots,L\}$ where $\nu_\ell^{(i)} = \frac{(1-\lambda)(1-\alpha_i)}{\eta  (2\alpha_i+\lambda(1-\alpha_i))d_{2,\ell}^{-\zeta}}$ for $i \in \{1,2\}$. Here, $X_2 \leq_{st} X_1$ if the following is true:
\begin{align}
	\sum\limits_{\ell=1}^L \exp{(-\theta_\ell z -\Tilde{\nu}_\ell\Tilde{\alpha}_2)} & \leq 	\sum\limits_{\ell=1}^L \exp{(-\theta_\ell z -\Tilde{\nu}_\ell\Tilde{\alpha}_1)},
	\label{order_alpha_eqn2}
\end{align}
where 
\begin{align}
    {\nu}_\ell^{(i)} = \Tilde{\nu}_\ell \Tilde{\alpha}_i, \  \Tilde{\nu}_\ell = \frac{(1-\lambda)}{\eta \times d_{2,\ell}^{-\zeta}} \ \text{and} \ \Tilde{\alpha}_i = \frac{1-\alpha_i}{2\alpha_i+\lambda(1-\alpha_i)}.
\end{align}
Upon further rearrangement, (\ref{order_alpha_eqn2}) can be re-written as follows
\begin{equation}
	\sum\limits_{\ell=1}^L \exp(-\theta_\ell)\left \lbrace \underbrace{\exp\left(-\Tilde{\nu}_\ell\Tilde{\alpha}_2 \right) - \exp\left(-\Tilde{\nu}_\ell\Tilde{\alpha}_1 \right)}_\text{Term 2}\right \rbrace \leq 0.
	\label{order_alpha}
\end{equation}

Further analysis shows that for, $\alpha_2<\alpha_1$, $\Tilde{\alpha}_2>\Tilde{\alpha}_1$. Hence, term 2 will be negative for all values of $\ell$ and hence the inequality in (\ref{order_alpha}) holds for all choices of $\Tilde{\nu_\ell}$ and $\theta_\ell$. Thus, we conclude that $X_2$ is stochastically smaller than $X_1$ when $\alpha_2<\alpha_1$. This means that the maximum e2e SNR increases with an increase in the TS factor i.e with increase in the time over which energy is harvested. However, note that we cannot choose $\alpha=1$, since this would mean that the whole time slot is utilised for energy harvesting and no time is allocated for information transfer. This means, practically we are constrained to choose a maximum value of $\alpha$ that still reserves time for information transmission both from the source to the relay and from the relay to the destination.

\subsubsection{Ordering with respect to $\sigma^2$} \label{order3}

Let $X_1$ and $X_2$ be the RVs representing the asymptotic maximum e2e SNR with the noise power $\sigma_1^2$ and $\sigma_2^2$ respectively. We further assume that $\sigma_1^2<\sigma_2^2$, RV $X_1$ has parameters $\{\theta_\ell^{(1)},\nu_\ell;\ell=1,\cdots,L\}$ and  $X_2$ has parameters  $\{\theta_\ell^{(2)},\nu_\ell;\ell=1,\cdots,L\}$ where $\theta_\ell^{(i)}=\frac{\sigma_i^2}{(1-\lambda)P_sd_{1,\ell}^{-\zeta}}$ for $i \in \{1,2\}$. Following the analysis similar to the case of $P_s$, we can infer that $X_2$ is stochastically smaller than $X_1$ in this case. 

\subsubsection{Ordering with respect to $\lambda$} \label{order4}

Let $X_1$ and $X_2$ be the RVs representing the asymptotic maximum e2e SNR with the PS factor $\lambda_1$ and $\lambda_2$ respectively. We further assume that $\lambda_1<\lambda_2$, RV $X_1$ has parameters $\{\theta_\ell^{(1)},\nu_\ell^{(1)};\ell=1,\cdots,L\}$ and  $X_2$ has parameters  $\{\theta_\ell^{(2)},\nu_\ell^{(2)};\ell=1,\cdots,L\}$ where 
\begin{align}
  \theta_\ell^{(i)} = \frac{\sigma_\ell^2}{(1-\lambda_i)P_s d_{1,\ell}^{-\zeta}} \ \text{and} \ \nu_\ell^{(i)} = \frac{(1-\lambda_i)(1-\alpha)\sigma_D^2}{\eta \sigma_\ell^2 (2\alpha+\lambda_i(1-\alpha))d_{2,\ell}^{-\zeta}}.  
\end{align}
Furthermore, we define
\begin{align}
     \theta_\ell^{(i)}=\frac{\Tilde{\theta_\ell}}{1-\lambda_i}, \   & \nu_\ell^{(i)}=\Tilde{\nu}_\ell\Tilde{\lambda}_i, \\
     \Tilde{\nu}_\ell = \frac{d_{2,\ell}^{\zeta}(1-\alpha)}{\eta} \ & \text{and} \  \Tilde{\lambda}_i = \frac{1-\lambda_i}{2\alpha+\lambda_i(1-\alpha)} \ \text{for} \ i \in \{1,2\}.
\end{align}
Here, $X_2 \leq_{st}  X_1$ if the following is true: 
\begin{equation}
    \sum\limits_{\ell=1}^L \underbrace{\exp\left(-\frac{\Tilde{\theta}_\ell z }{1-\lambda_2}\right)}_{\text{Term 3}}  \underbrace{\exp(-\Tilde{\nu}_\ell\Tilde{\lambda_2})}_{\text{Term 4}}  - \underbrace{\exp\left(-\frac{\Tilde{\theta}_\ell z }{1-\lambda_1}\right)}_{\text{Term 5}} \underbrace{\exp(-\Tilde{\nu}_\ell\Tilde{\lambda_1})  \leq 0}_{\text{Term 6}}. 
    \label{order_lambda}
\end{equation}
Unlike the case of TS factor $\alpha$, the ordering with respect to $\lambda$ does not remain the same for all values of $\theta_\ell$ and $\nu_\ell$, $\ell=1,\cdots,L$. Here, note that for $\lambda_1,\lambda_2 \in (0,1)$ and $\lambda_1<\lambda_2$, term 3 is smaller than term 5 and term 4 is larger than term 6. Hence, the sign of left-hand-side (LHS) of (\ref{order_lambda}) will depend on whether Term 3 (or Term 5) dominates Term 4 (or Term 6) or otherwise, in each of the $\ell$ sum terms. Note that in the high SNR scenario, $\Tilde{\theta}_\ell$ will be small since $\Tilde{\theta}_\ell$ is inversely proportional to $\gamma_s=\frac{P_s}{\sigma^2}$. Now if $\lambda_i$ is not very close to 1, the product terms in (\ref{order_lambda}) will be dominated by term 4 and term 6. If $\lambda_i$ is close to 1, the denominator of the exponent of terms 3 and 5 will tend to zero. Similarly, for the low SNR regime $\Tilde{\theta}_\ell$ will be large and hence smaller values of $\lambda_i$ will increase the value of the product terms. However, for all the cases in between the high and low SNR values we cannot have a general conclusion about the ordering of $\gamma_{e2e,max}$ with respect to the variations in the PS factor $\lambda$.
\par Note that it is not easy to derive the above inferences using the exact expression of the maximum CDF. The simple form of the asymptotic maximum CDF was instrumental in simplifying the above analysis. Given that we have established the ordering of $\gamma_{e2e, max}$ with respect to variations in $P_s,\alpha$ and $\sigma^2$, we can extend this to the case of asymptotic ergodic capacity by making use of the following result from the theory of stochastic ordering.

\begin{lemma} \label{sto_order}
RV $X$ is stochastically less than or equal to RV $Y$ if and only if the following holds for all increasing functions $\phi(.)$ for which the expectations exist : 
\begin{equation}
    \mathbb{E}[\phi(X)] \leq \mathbb{E}[\phi(Y)].
\end{equation}
\end{lemma}
The above lemma is discussed in detail in chapter 1 of \cite{shaked2007stochastic}. Using Lemma \ref{sto_order} we can easily extend the ordering results in section \ref{order1}.\ref{order4} to the ordering of asymptotic ergodic capacity $C_{e2e,max}$. This in turn allows us to make inferences about the changes in the asymptotic ergodic capacity with respect to variations in the system parameters easily. Note that such observations are otherwise difficult to be derived directly from the integral expression for ergodic capacity given in (\ref{asym_erg_intg}).

\section{Optimal TS and PS parameter} \label{opti_probsection}

Note that the statistics of the e2e SNR depends on the choice of TS and PS factors. In this section, we discuss one possible method to choose the optimal TS and PS factor for (a) minimising the outage probability and (b) maximising the ergodic capacity. 
\subsection{Minimising outage probability}
First, we consider the problem of choosing the optimal TS and PS factor that minimises the outage probability at the destination node. More precisely, we look at the following optimisation problem : 
\begin{subequations}
\begin{alignat}{2}
& \min_{\alpha,\lambda}        &\qquad& {\exp(-\sum\limits_{\ell=1}^L \exp{(-\theta_\ell \gamma -\nu_\ell)})}\label{eq:optProb}\\
&\text{subject to} &      & 0 \leq \alpha \leq \alpha^{max},\label{eq:constraint1}\\
&                  &      & 0 \leq \lambda \leq \lambda^{max},\label{eq:constraint2}
\end{alignat}
\label{opti_prob}
\end{subequations}
where $\alpha^{max}$ and $\lambda^{max}$ are the maximum values of $\alpha$ and $\lambda$, feasible within the hardware constraints of the system. Here, the objective function (\ref{eq:optProb}) represents the outage probability at the destination for TS and PS factor $\alpha$ and $\lambda$ respectively. Note that, $\alpha=1$ and $\lambda=1$ would  mean that the relays only harvest energy and does not transmit any information, and cannot be a valid choice. 
By monotonicity of the logarithm function, the optimal solution for the above optimisation problem would remain unchanged even if the objective function is replaced with log of the outage probability. Hence, (\ref{eq:optProb}) can now be replaced with
\begin{equation}
    \min_{\alpha,\lambda}        \qquad {-\sum\limits_{\ell=1}^L \exp{(-\theta_\ell \gamma -\nu_\ell)}}.\label{eq:optProb_log}
\end{equation}The minimisation of the 
above objective is equivalent to the maximisation of the negative of the same \cite{boyd2004convex}. Thus, the optimisation problem in (\ref{opti_prob}) can be rewritten as follows: 
\begin{subequations}
\begin{alignat}{2}
& \max_{\alpha,\lambda}        &\qquad& {\sum\limits_{\ell=1}^L \exp{(-\theta_\ell \gamma -\nu_\ell)}}\label{eq:optProb1}\\
&\text{subject to} &      & 0 \leq \alpha \leq \alpha^{max},\label{eq:constraint11}\\
&                  &      & 0 \leq \lambda \leq \lambda^{max}.\label{eq:constraint21}
\end{alignat}
\label{opti_prob1}
\end{subequations}
The objective function in (\ref{eq:optProb1}) is neither convex nor concave. So the next step is to see what optimisation algorithm we can use to identify the local maximum of the above objective function. In the following paragraph, we demonstrate how the stochastic ordering results of Section III.B can be used to simplify the above optimisation problem.
\par From the ordering results in Section III.B, we know that the objective function increases with an increase in $\alpha$. Hence, the optimal TS factor ($\alpha^*$) for the optimisation problem in (\ref{opti_prob1}) will be $\alpha^{max}$. Thus, the original bi-variate optimisation problem can now be solved as an uni-variate optimisation problem. Next, to identify the optimal PS factor that minimises the probability of outage, we have to search for the local maximum of the objective function (\ref{eq:optProb1}) evaluated at $\alpha=\alpha^*$. There are a number of algorithms for finding the optima of non-convex, nonlinear optimisation problems. Here, we propose to use the method of sequential quadratic programming (sqp) to find the optimum solution \cite{websqp}, which under specific conditions is proven to demonstrate faster convergence as compared to algorithms like the interior point method. Furthermore, we can make use of the stochastic ordering results in high and low SNR regimes to do clever initialisation of the sqp algorithm and thus accelerate the convergence of the algorithm. More details regarding the initialisation of the algorithm are presented along with the simulation results in Section \ref{Simulation}.

\subsection{Maximising ergodic capacity}
In scenarios where the e2e capacity is more important than the outage probability at the destination, the following optimisation problem can be solved to find the optimal TS and PS factor that maximises the asymptotic ergodic capacity.
\begin{subequations}
\begin{alignat}{2}
& \max_{\alpha,\lambda}        &\qquad& C_{e2e,max}\label{eq:optProb_ce2e}\\
&\text{subject to} &      & 0 \leq \alpha \leq \alpha^{max},\label{eq:constraint1_ce2e}\\
&                  &      & 0 \leq \lambda \leq \lambda^{max},\label{eq:constraint2_ce2e}
\end{alignat}
\label{opti_prob_ce2e}
\end{subequations}
where $\alpha^{max}$ and $\lambda^{max}$ are the maximum values of TS and PS factors feasible within the hardware constraints of the system. Now, similar to the previous sub-section we can use the stochastic ordering results to simplify the above optimisation problem. Using Lemma 2 from Section III.B, we arrive at the conclusion that the asymptotic ergodic capacity $C_{e2e,max}$ increases with an increase in the TS factor $\alpha$. Hence, the optimal value of TS factor for the above optimisation problem is given by $\alpha^*=\alpha^{max}$. The bi-variate optimisation problem in (\ref{opti_prob_ce2e}) can thus be solved using the uni-variate optimisation problem given in (\ref{opti_prob_ce2e_new}) where $C_{e2e,max}(\alpha^{max},\lambda)$ is the asymptotic ergodic capacity evaluated at $\alpha=\alpha^{max}$. Note that the system hardware constraints decide the resolution with which power splitting can be implemented and hence decide the possible choices for $\lambda$. Thus, we replace the constraint in (\ref{eq:constraint2_ce2e}) with the constraint in (\ref{eq:constraint1_ce2e_new}) where $\Lambda$ is a finite set of all the possible values of $\lambda$. 
\begin{subequations}
\begin{alignat}{2}
& \max_{\lambda}         & C_{e2e,max}(\alpha^{max},\lambda)\label{eq:optProb_ce2e_new}\\
&\text{subject to}  \quad &      \lambda \in  \Lambda.\label{eq:constraint1_ce2e_new}
\end{alignat}
\label{opti_prob_ce2e_new}
\end{subequations}
Since a simple closed form expression of $C_{e2e,max}$ is not available, we propose to use a simple line search algorithm to find the optimal value of $\lambda$. The integral expression for $C_{e2e,max}$ can be easily evaluated using numerical integration methods, for example using the \textit{NIntegrate} method available in \textit{Mathematica}.
\subsection{Summary of key insights}
We now present the key insights from the analysis in Section III-IV below.
\begin{itemize}
\item {We prove that the distribution of an appropriately normalised maximum e2e SNR RV converges to the distribution of the Gumbel RV, where the maximum is evaluated over a set of i.n.i.d. RVs.}
\item {We applied these results and studied the distribution of the maximum e2e SNR RV, $\max \{\gamma_{e2e,\ell}\}_{\ell=1}^L$ for large $L$ and the corresponding CDF is given by}
    \begin{align}
        & F_{\gamma_{e2e,max}}(\gamma) = \exp(-u(\gamma)) \  \ \text{where} \label{unnorm_dist}\\ 
    & u(\gamma)  = \sum\limits_{\ell=1}^L \exp{(-\theta_\ell \gamma -\nu_\ell)},  \  \theta_\ell = \frac{\sigma_\ell^2}{(1-\lambda)P_s d_{1,\ell}^{-\zeta}} \ \text{and} \ \nu_\ell = \frac{(1-\lambda)(1-\alpha)\sigma_D^2}{\eta \sigma_\ell^2 (2\alpha+\lambda(1-\alpha))d_{2,\ell}^{-\zeta}}.
    \end{align}
\item {The above distribution in (\ref{unnorm_dist}) is both easy to evaluate and simple to analyse when compared to the exact distribution of the maximum RV which is given by
    \begin{equation}
      F_{\gamma_{e2e,max}^L}(\gamma) = \prod\limits_{\ell=1}^L 1-\theta_\ell \sum\limits_{k=0}^ \infty \frac{(-\nu_\ell)^k}{k!} \gamma E_k(\theta_\ell \gamma),
    \end{equation}where $E_k(.)$ is the exponential integral function.} 
\item {The distribution in (\ref{unnorm_dist}) is close to the exact distribution of the maximum even for moderate values of $L$. Please see simulations in Section V for more details. Hence, this asymptotic maximum order statistics can be used for performance analysis and resource planning of the dual hop SWIPT CR system.} 
\item {We also derived simple integral expressions for the ergodic capacity and achievable throughput of the CR system using the distribution of $\gamma_{e2e,max}$. Note that deriving the expressions for the achievable throughput and outage capacity using the exact maximum statistics would have been an computationally intensive task. }
\item {Using the results derived, we establish the stochastic ordering of $\gamma_{e2e,max}$ with respect to the variations in source transmit power, noise variance, TS factor and PS factor. Arriving at such conclusions using simple algebra would not have been possible using the exact order statistics.}
\item {Finally, the utility of our results is emphasised by demonstrating how the CDF expressions and the stochastic ordering results could be used to select the optimal TS and PS factor for minimising the outage probability and maximising the ergodic capacity at the destination. }
\end{itemize}

\section{Simulation Results} \label{Simulation}
\subsection{Motivation for the analysis of extreme statistics of i.n.i.d sequences of RVs}
 \par Before we begin with the detailed discussion of the simulation results for the i.n.i.d. scenario, we motivate the importance of the analysis through an example. In most of the practical scenarios, we are interested in analysing the maximum or the minimum statistics over a sequence of i.n.i.d. RVs. However, several works in the literature assume them to be i.i.d. RVs and the right method for approximating a sequence of i.n.i.d. RVs with a sequence of i.i.d. RVs is an interesting problem in itself.
\begin{figure}[t]
	\centering
		\begin{minipage}[t]{0.45\textwidth}
		\includegraphics[scale=0.4]{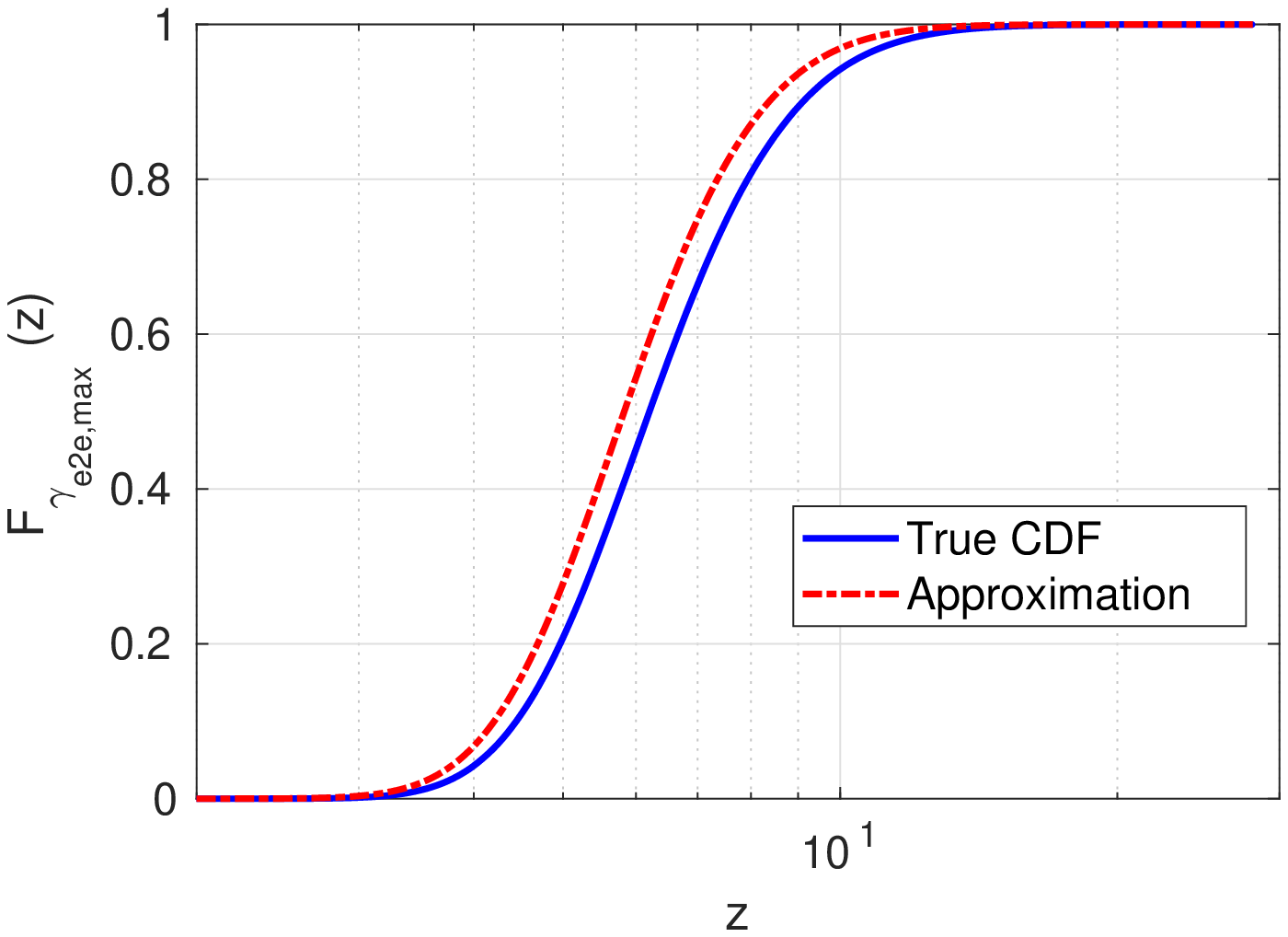}
		\caption{ CDF of $\gamma_{e2e,max}$ with $\theta_\ell=1;1\leq \ell\leq \frac{L}{2}$ and $\theta_\ell=3;\frac{L}{2} < \ell\leq L$}
		\label{e2e_cdf_t2}
	\end{minipage}
	\begin{minipage}[t]{0.45\textwidth}
		\includegraphics[scale=0.4]{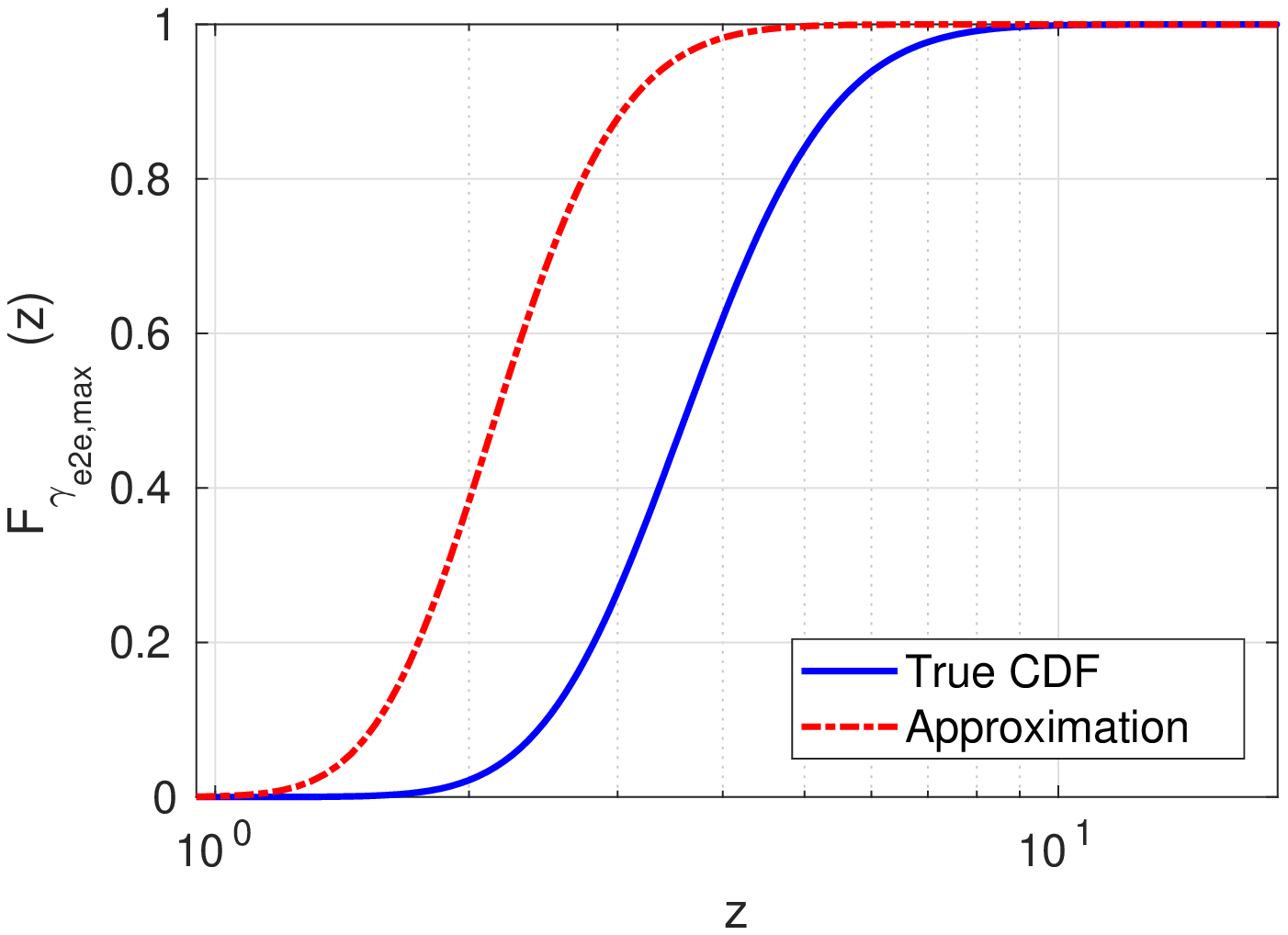}
		\caption{CDF of $\gamma_{e2e,max}$ with $\theta_\ell \sim \text{Uniform Distribution}(1,3)$}
		\label{e2e_cdf_t1}		
	\end{minipage}
\end{figure}
In the following two experiments, we consider $L$ i.n.i.d. RVs of the form $\gamma_{e2e,\ell}$ each with CDF given by Lemma \ref{cdf_e2e_thm}. We choose the parameters as $\nu_\ell=\nu=0.2 \ \forall \ \ell \in \{1,\cdots,L\}$ and $\theta_{\ell}$ to be non-identical across the RVs with values chosen from the interval $[1,3]$. Here, the solid curve in Fig \ref{e2e_cdf_t2} and Fig \ref{e2e_cdf_t1} shows the empirical CDF of the maximum of such RVs over sequences of length $L=64$. Now, if we approximate this sequence of RVs with an sequence of i.i.d. RVs each with $\theta_\ell=\theta$ chosen to be the mean of the non-identical parameters i.e $\theta=\frac{1}{L}\sum\limits_{\ell=1}^L \theta_\ell$, the asymptotic distribution of the maximum will be a Gumbel distribution with location and scale parameters as discussed in Section \ref{asymp_derive_sec}. This approximate CDF generated using the theoretical distribution function of the Gumbel RV is plotted as the red dashed curve in the figures. Here we notice that the above approximation is a good choice for the first case whereas, in the second case, the theoretical approximation and the true simulated CDF are not close. In other words, for i.n.i.d. RVs, an approximation which works well for one set of values may be poor for another set of values. Hence, directly deriving the asymptotic distribution of the maximum of the sequence of i.n.i.d. RVs can give more accurate inferences about system performance and utility.
 \begin{figure}[t]
	\centering
		\begin{minipage}[t]{0.45\textwidth}
		\includegraphics[scale=0.45]{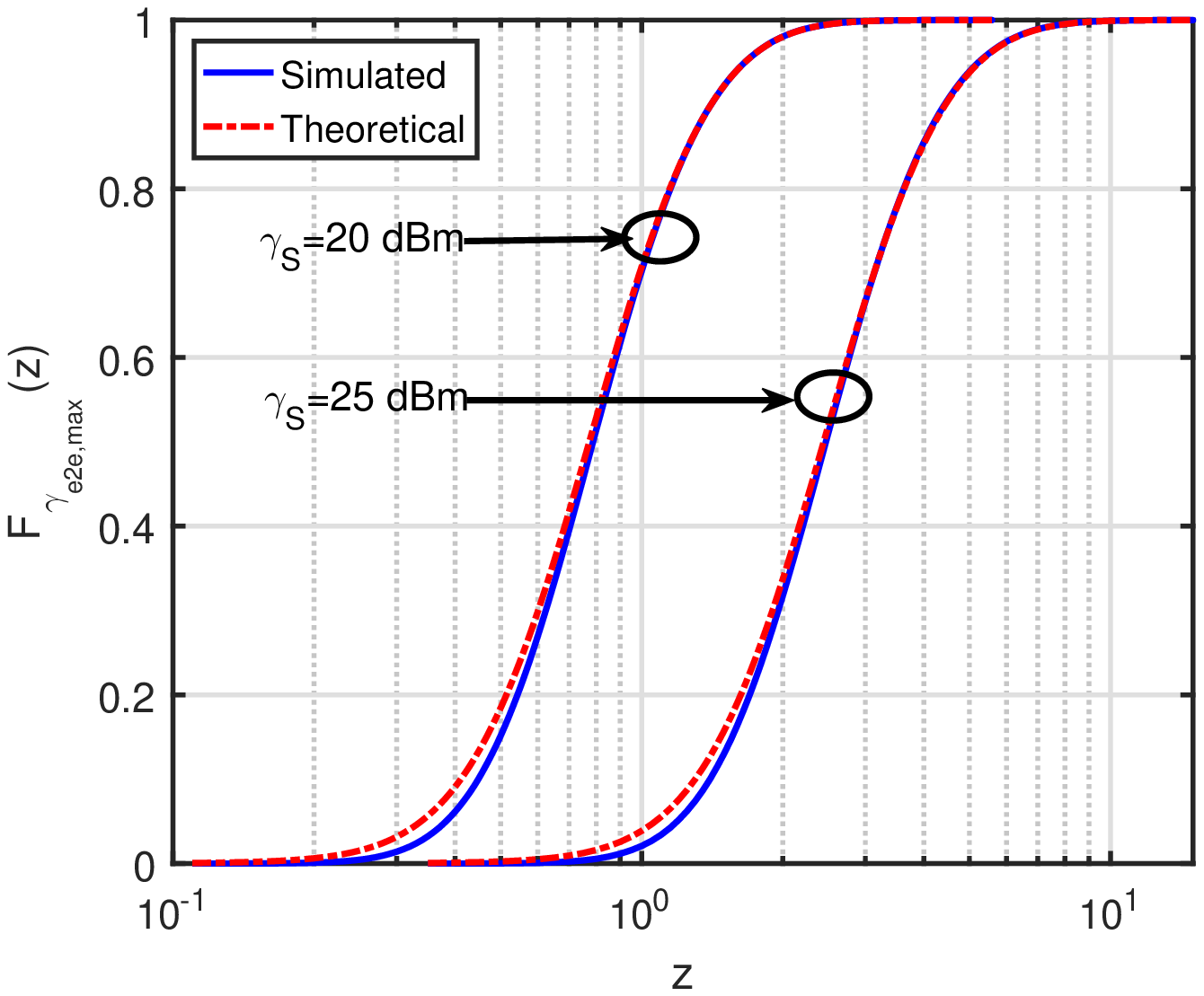}
		\caption{CDF of $\gamma_{e2e,max}$ with $L=15$.}
		\label{cdf_L_15}
	\end{minipage}
	\begin{minipage}[t]{0.45\textwidth}
		\includegraphics[scale=0.45]{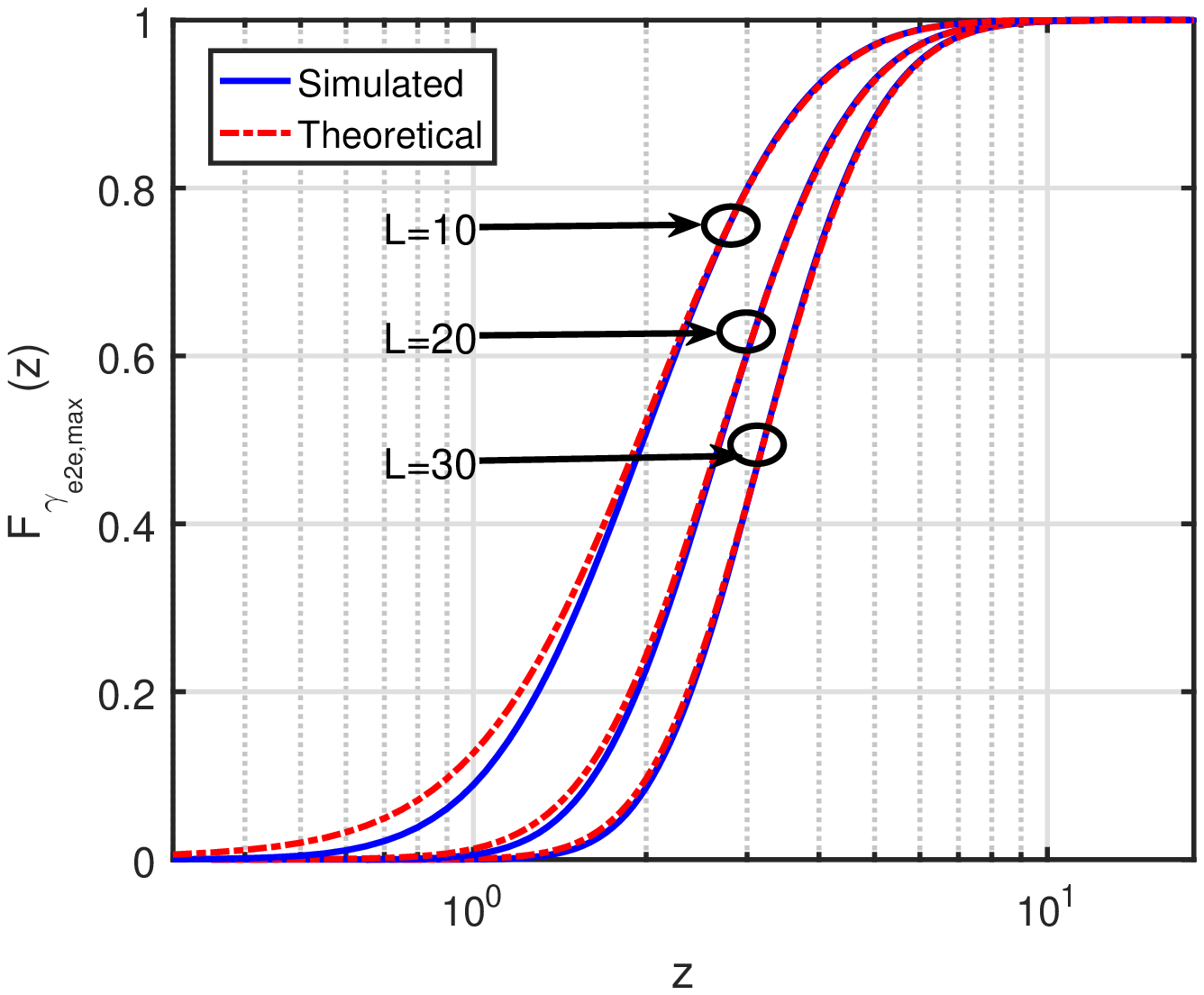}
		\caption{CDF of $\gamma_{e2e,max}$ with $\gamma_S=25 \ dBm$.}
		\label{cdf_gammas}		
	\end{minipage}
 \end{figure}
\subsection{Simulations for the results in Section \ref{asymp_derive_sec}}
 Next, we present results of simulation experiments to demonstrate the validity of the asymptotic distribution derived in Section III. Here, we choose the noise power to be identical at all the relays as well as the destination and we define $\gamma_s:= \frac{P_s}{\sigma^2}$ where $\sigma^2$ is the noise power. Throughout the simulations, we have chosen $\gamma_s=25$dBm, $\eta=0.9$, $L=20$, $\alpha=0.3$, $\lambda=0.4$ and $\gamma_{th}=1$ dB unless stated otherwise. Furthermore, we assume that the straight line distance between source and destination is normalised to unity. The distance from the source to the relays and relays to the destination are then uniformly chosen from intervals $(0.5,0.8)$ and $(0.5,0.7)$ respectively. Here, Fig \ref{cdf_L_15} and Fig \ref{cdf_gammas} show the simulated and theoretical CDF of $\gamma_{e2e,max}$ for different values of $L$ and $\gamma_s$. From the figures, we see that the asymptotics hold good even when the maximum SNR is evaluated over a small number of relays, $L$. Furthermore, we can see that the convergence of the exact distribution of the maximum to the asymptotic distribution improves with an increase in $L$.  Fig \ref{the_outage_cap} shows the theoretical values of outage capacity $C_{e2e,max}^{out}$ for different combinations of TS and PS factors. Here we notice that the outage capacity decreases significantly with increasing $\lambda$. 
\subsection{Simulations for the results in Section \ref{erg_cap_out_cap}}
{Next, we present simulation results to validate the convergence of the ergodic capacity to the proposed value of asymptotic ergodic capacity. Fig \ref{L_vs_rate_1} shows the simulated and theoretical values of achievable throughput for different values of $\gamma_s$ and $L$. We see that the simulated and theoretical values are in good agreement for all values of $L>10$. This validates the utility of the asymptotic results in many system planning problems where achievable throughput is the factor of interest. The variation in the theoretical values of achievable throughput for different combinations of TS and PS factors are given in Fig \ref{the_alpha_lambda_vary_rate_test3_th_1db}. It is to be noted that the achievable throughput does not show the same trend as the outage capacity but decreases with increase in $\alpha$ beyond a certain value. }

\begin{figure}[t]
    \centering
		\begin{minipage}[t]{0.5\textwidth}
		\includegraphics[scale=0.5]{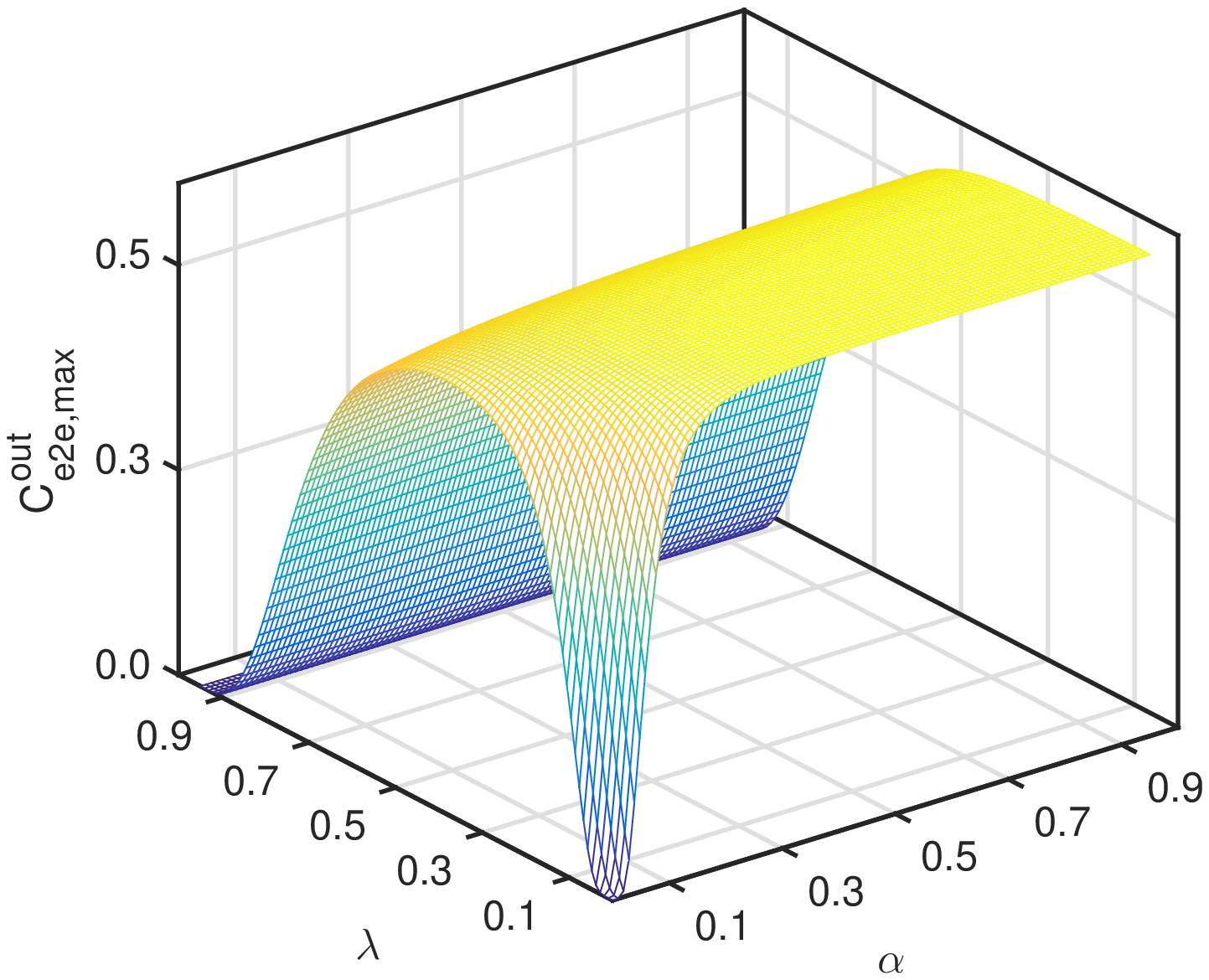}
		\caption{Outage capacity (in bits/sec/Hz) as a function of $\alpha$ and $\lambda$ with $\gamma_{th}=1 \ dB$.}
		\label{the_outage_cap}
	\end{minipage}
	\begin{minipage}[t]{0.45\textwidth}
		\includegraphics[scale=0.5]{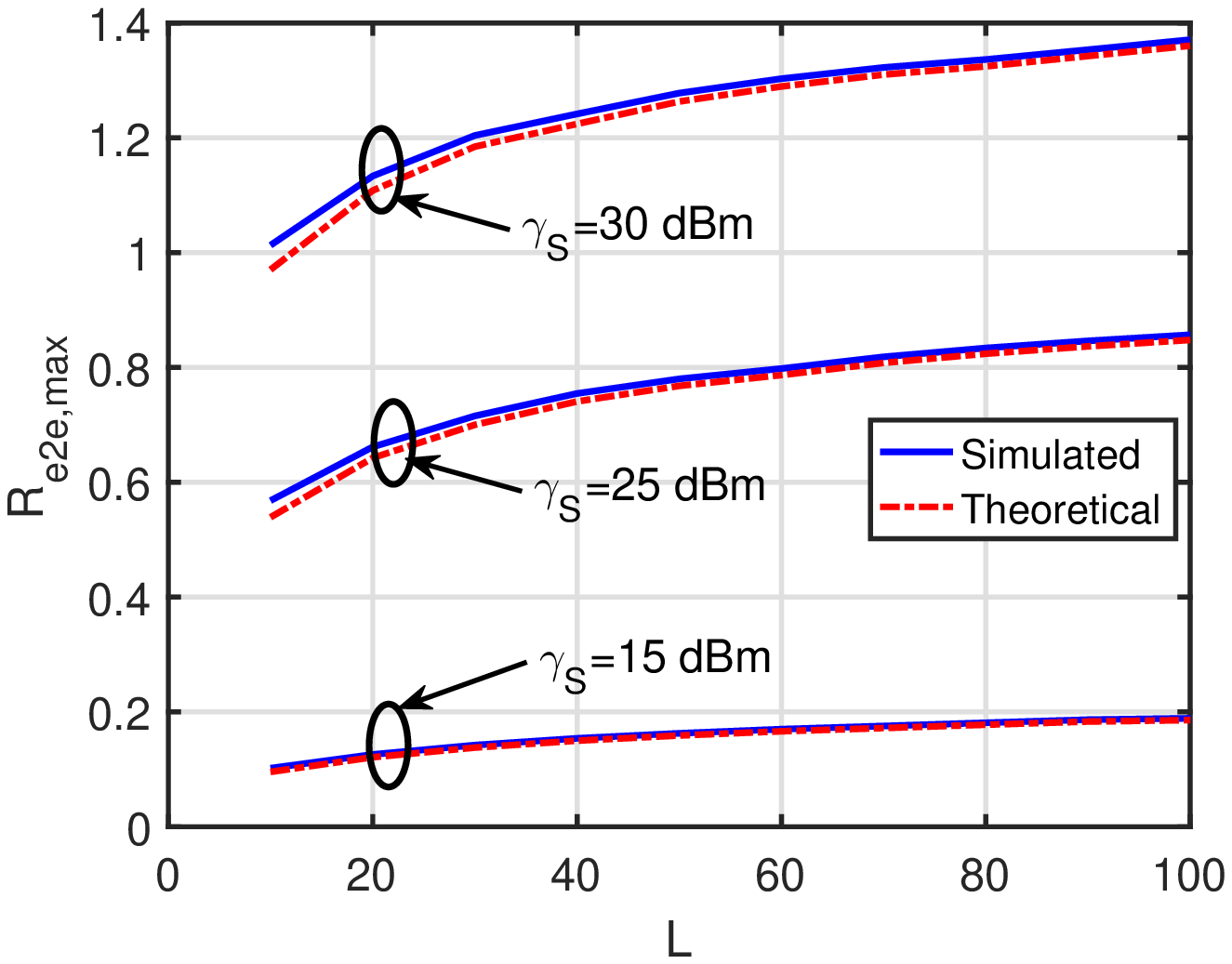}
		\caption{$R_{e2e,max}$ (in bits/sec/Hz) as a function of $L$.}
		\label{L_vs_rate_1}	
	\end{minipage}
 \end{figure}
 
\begin{figure}[t]
	\centering
	\begin{minipage}[t]{0.45\textwidth}
		\includegraphics[scale=0.5]{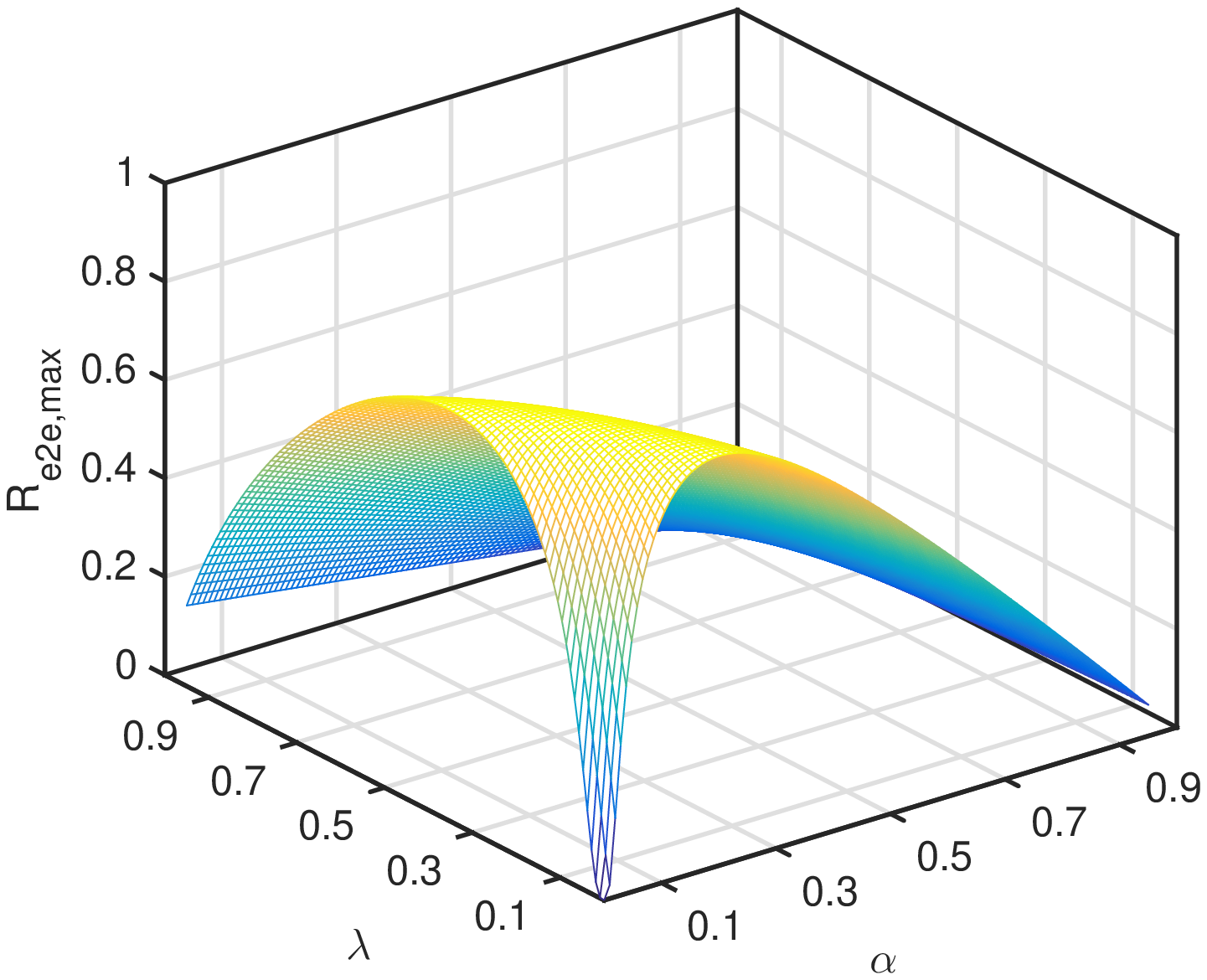}
		\caption{$R_{e2e,max}$ (in bits/sec/Hz) as a function of $\alpha$ and $\lambda$ with $\gamma_{s}=25 \ dBm$.}
		\label{the_alpha_lambda_vary_rate_test3_th_1db}	
	\end{minipage}
	\begin{minipage}[t]{0.45\textwidth}
		\includegraphics[scale=0.5]{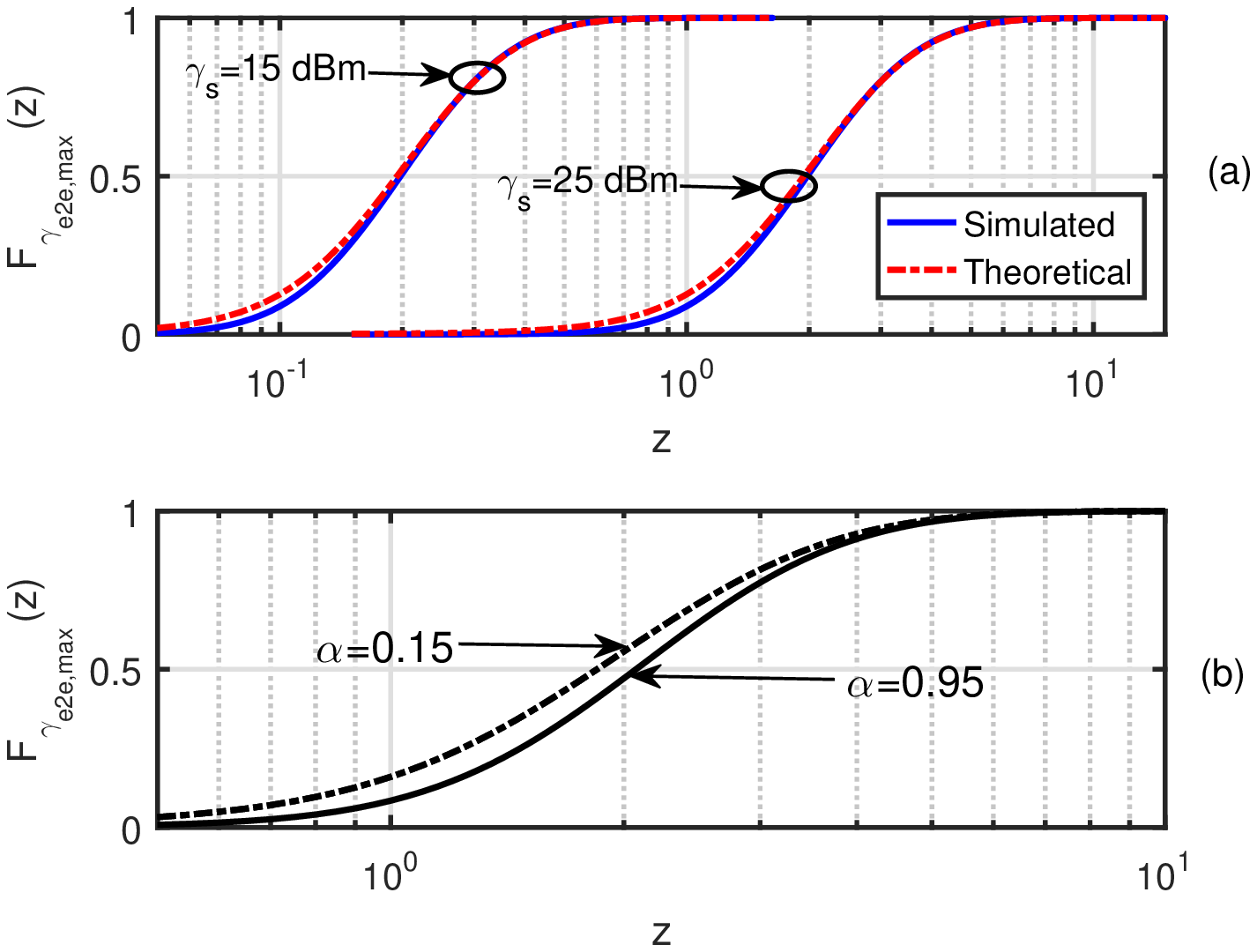}
		\caption{CDF of $\gamma_{e2e,max}$.}
		\label{order_test2}
	\end{minipage}%
\end{figure}
\subsection{Simulations for the results in Section \ref{ordering_e2esnr} and \ref{opti_probsection}}
{The ordering results in Section \ref{ordering_e2esnr} are verified in Fig \ref{order_test2} for $L=15$. For clarity in presentation, we have plotted only the theoretical curves of CDF in Fig \ref{order_test2} (b). Next, in Fig \ref{low_SNR_opti_soln1} and \ref{high_SNR_opti_soln1}, we demonstrate the solutions for the optimisation problem to choose the optimal TS and PS factors. Here the optimal solutions ($\alpha^*,\lambda^*$) are shown using a red star in the figure. In Fig \ref{low_SNR_opti_soln1} we show the log of outage probability and the corresponding choice of optimal $\alpha$ and $\lambda$ for $\gamma_s=4$ dBm and a threshold of $\gamma_{th} = 15 $ dBm. As discussed in Section \ref{opti_probsection}, for a low SNR scenario only one factor of each of the product terms in (\ref{order_lambda}) dominates and smaller values of $\lambda$ will increase the objective and hence decrease the outage. Hence, we propose that $\lambda=0$ will be a good initialisation. In fact, for the very low SNR scenario, we observe that the optimal choice corresponds to the TS relaying protocol. Also, it was observed that in these cases initialising $\lambda=0$ for the sqp algorithm reduces the number of iterations by half when compared to the number of iterations required for convergence when the initialisation is $\lambda=1$. This emphasises the utility of our ordering results and further reiterates the fact that the right initialisation can ensure faster convergence. }

\begin{figure}[t]
	\centering
	\begin{minipage}[t]{0.45\textwidth}
		\includegraphics[scale=0.45]{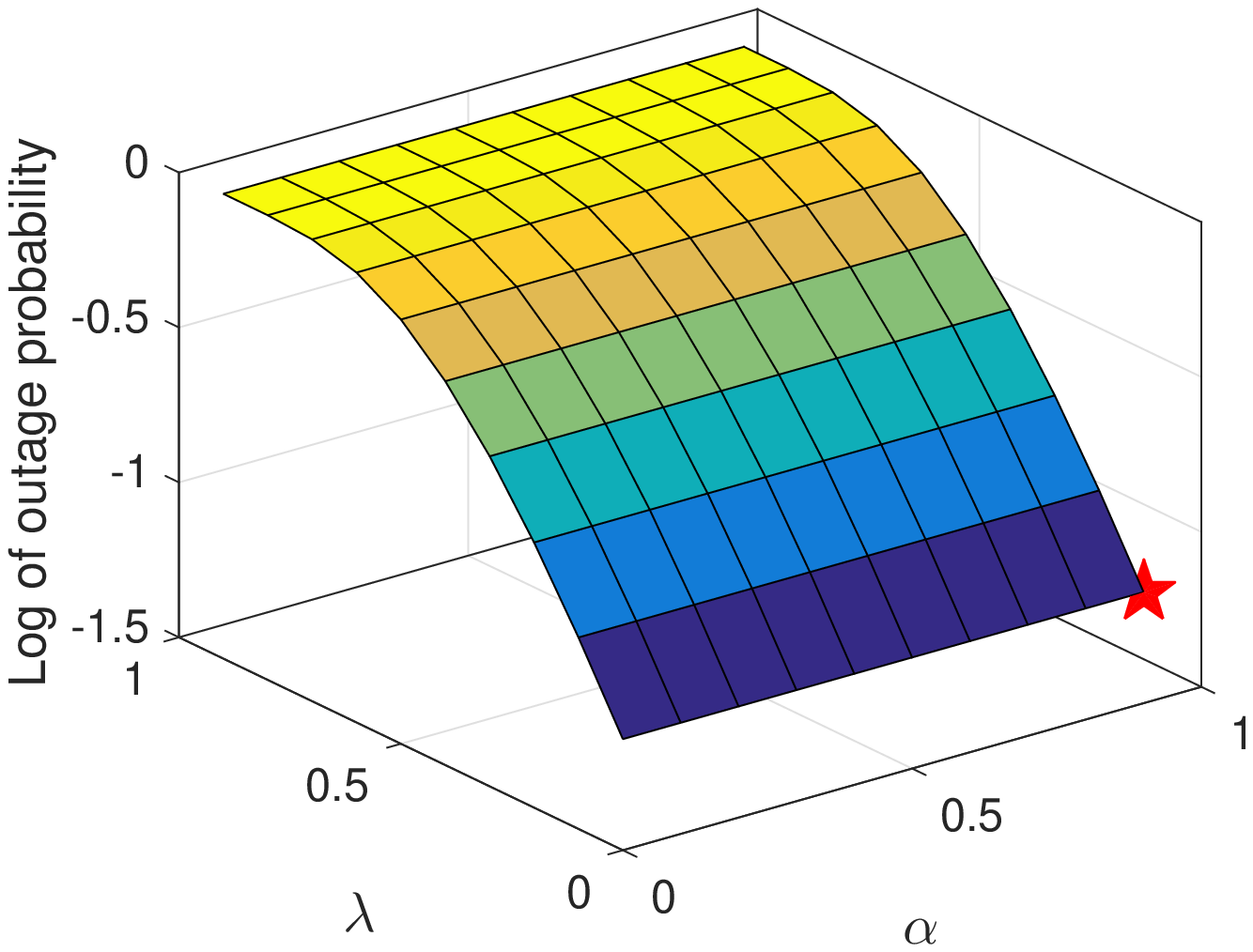}
		\caption{Log of outage probability as a function of $\alpha$ and $\lambda$ with $\gamma_s=4$ dBm. ($\alpha^*=0.9.\lambda^*=0$)}
		\label{low_SNR_opti_soln1}	
	\end{minipage}
	\begin{minipage}[t]{0.45\textwidth}
		\includegraphics[scale=0.45]{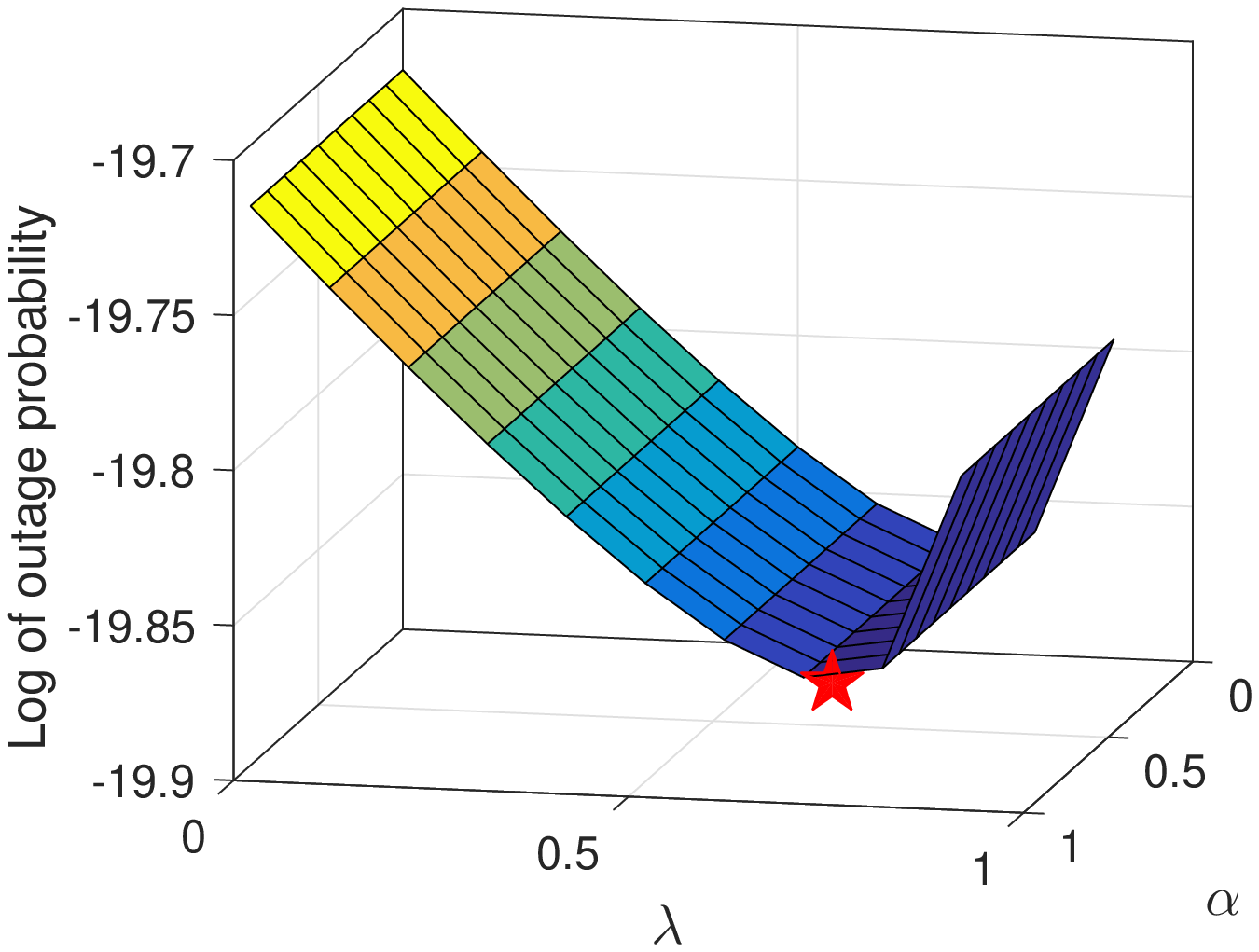}
		\caption{Log of outage probability as a function of $\alpha$ and $\lambda$ with $\gamma_s=40$ dBm. ($\alpha^*=0.9.\lambda^*=0.74$)}
		\label{high_SNR_opti_soln1}
	\end{minipage}%
\end{figure}
 
{Next, in Fig \ref{high_SNR_opti_soln1} we show another example of the optimisation problem for $\gamma_s=40$ dBm and the same threshold of $\gamma_{th} = 15 $ dBm. From the stochastic ordering results, we know that here optimal $\lambda$ can be away from zero but not equal to one as well. Hence, we decide that $\lambda=0.5$ can be a good initialisation. Here also we observe slightly faster convergence with this initialisation as compared to any other initialisation away from the optimal solution. For any other SNR scenario, we propose to use $\lambda=0.5$ as the initialisation since the solution has to be between the above two cases.}\\ 
Fig \ref{the_alpha_lambda_ce2e_max} shows variation of the asymptotic ergodic capacity with respect to the variations in the TS and PS factors for $\gamma_s=25$dBm. The red star represents the optimal values of $\alpha$ and $\lambda$ that maximises the asymptotic ergodic capacity according to the optimisation problem in Section IV-B. Similarly, we can solve the optimisation problem to maximise the asymptotic achievable throughput. Note that for the case of asymptotic achievable throughput, we cannot make a general conclusion regarding the sign of $\frac{\partial R_{e2e,max}}{\partial \alpha}$. Hence, we propose a 2-D grid search over the set of all possible values of TS and PS factors to arrive at the optimal values that maximise $R_{e2e,max}$.  Fig \ref{the_alpha_lambda_re2e_max} shows the simulation results for one such optimisation problem solved using the grid search method.
\begin{figure}[t]
	\centering
	\begin{minipage}[t]{0.4\textwidth}
		\centering
		\includegraphics[scale=0.45]{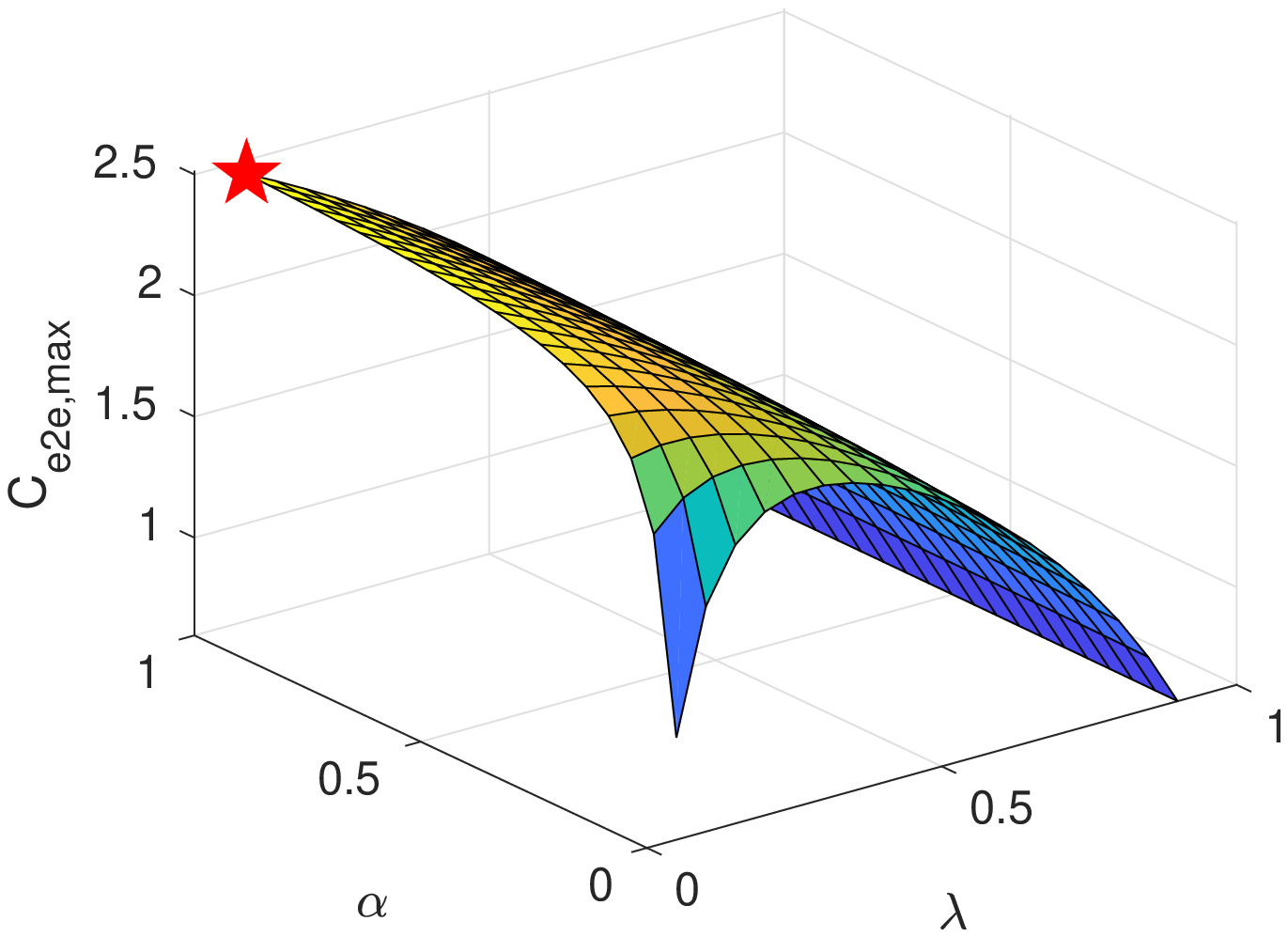}
		\caption{Variation in $C_{e2e,max}$ (in bits/sec/Hz) versus TS ($\alpha$) and PS ($\lambda$) factors for $\gamma_s=25$ dBm. ($\alpha^*=0.9.\lambda^*=0$)}
		\label{the_alpha_lambda_ce2e_max}
	\end{minipage}
	\begin{minipage}[t]{0.5\textwidth}
		\includegraphics[scale=0.45]{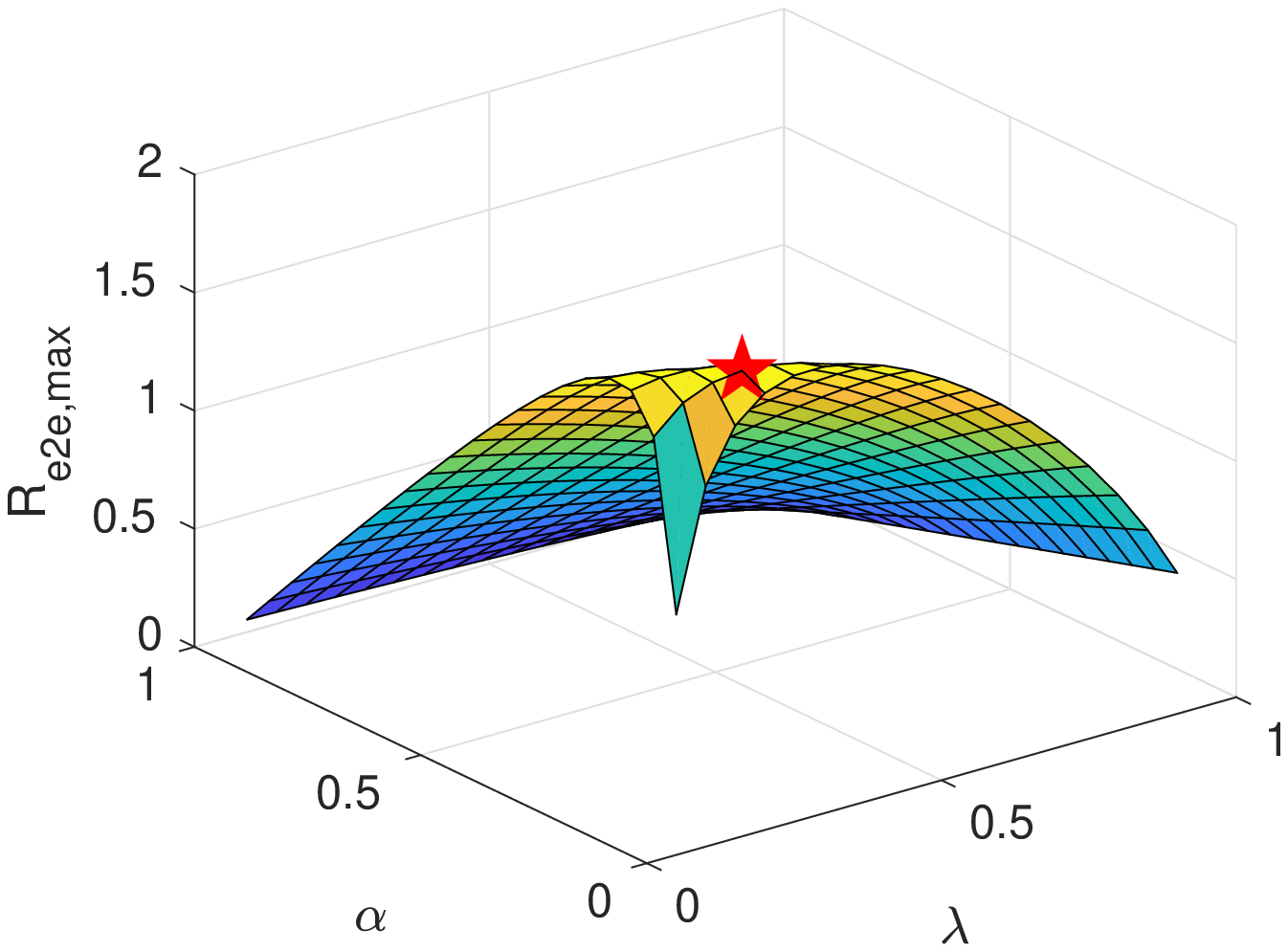}
		\caption{Variation in $R_{e2e,max}$ (in bits/sec/Hz) versus TS ($\alpha$) and PS ($\lambda$) factors for $\gamma_s=25$ dBm. ($\alpha^*=0,\lambda^*=0.2$)}
		\label{the_alpha_lambda_re2e_max}		
	\end{minipage}
 \end{figure}
\par For the outage-optimal TS and PS factors, we compare the outage probability for the three EH protocols in Fig \ref{compare_3_out}. Here we observe that the hybrid protocol and the TS protocol achieves identical performance in terms of the outage probability at the destination. Next, Fig  
\ref{compare_3_ce2e} compares the performance of the three protocols when the TS and PS factor that maximises the asymptotic ergodic capacity is chosen for the simulation. Here also we observe that the TS protocol achieves performance very close to the hybrid protocol. Thus, in a scenario where the system hardware constraints allows only PS protocol, we will have to use higher transmit power to achieve the same performance achievable via systems with the TS or hybrid protocol.

\begin{figure}[t]
	\centering
	\begin{minipage}[t]{0.4\textwidth}
		\includegraphics[scale=0.45]{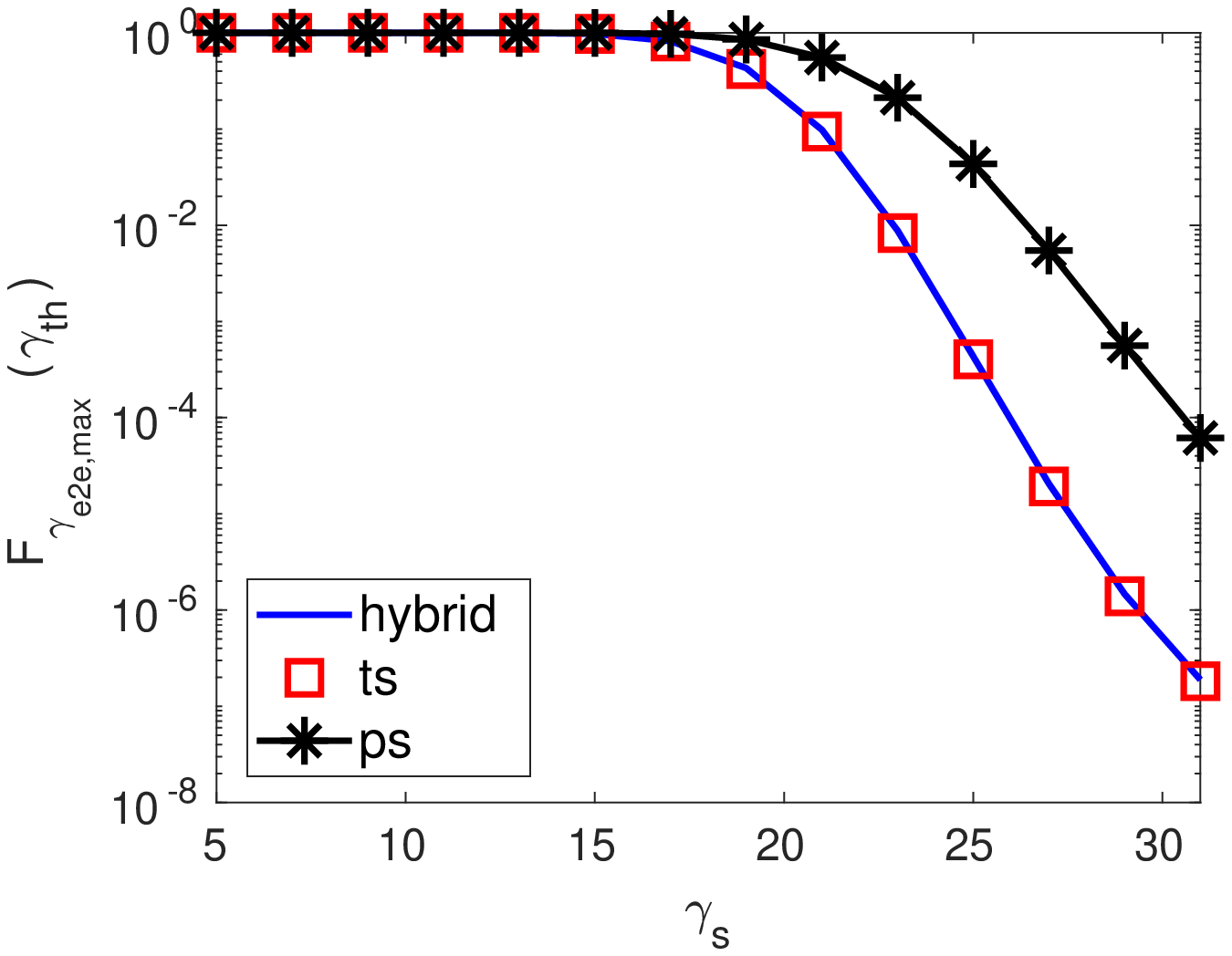}
		\caption{Outage probability vs SNR of the three EH protocols for outage-optimal TS and PS factors.  }
		\label{compare_3_out}	
	\end{minipage}
 \begin{minipage}[t]{0.4\textwidth}
		\centering
	\includegraphics[scale=0.45]{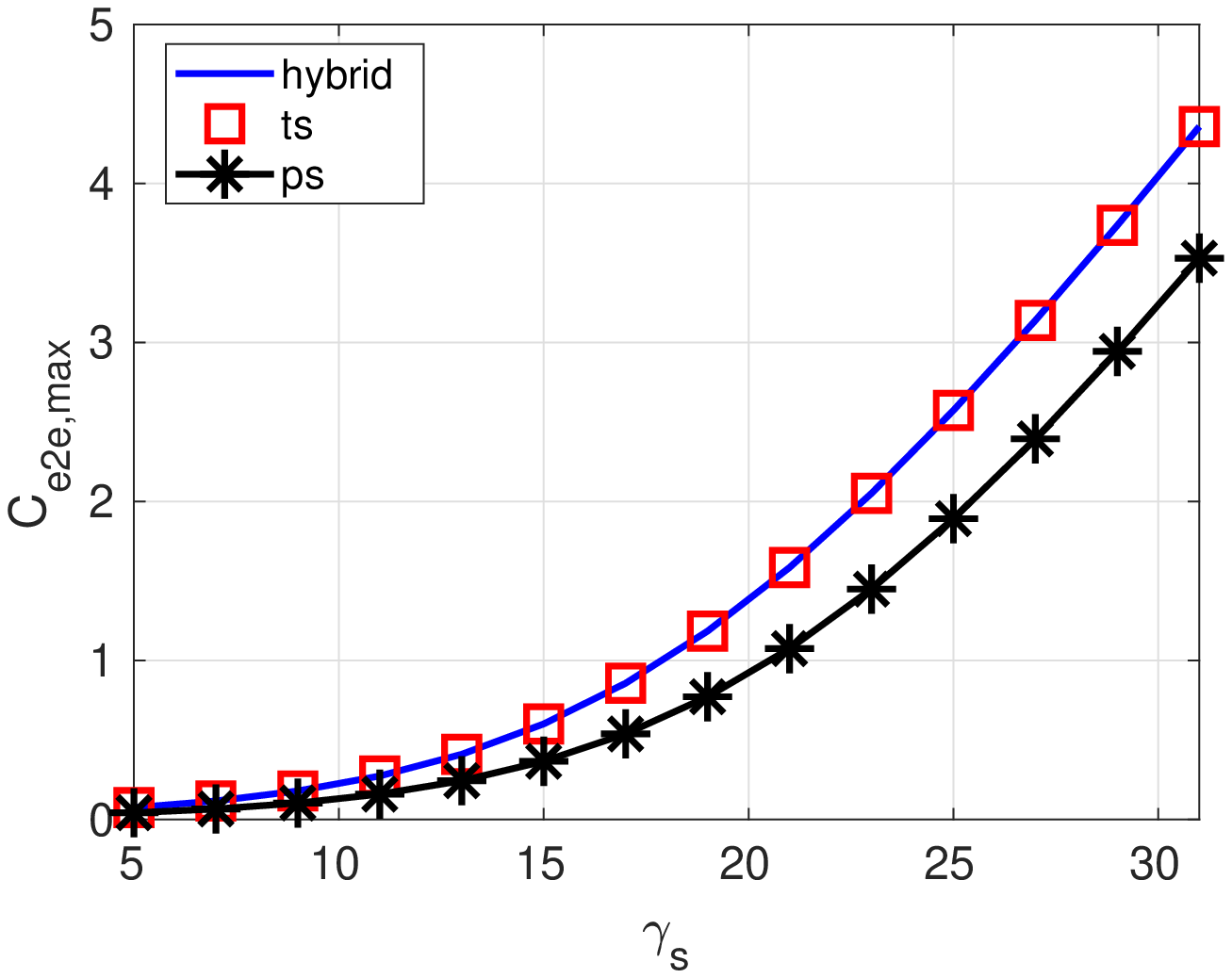}
		\caption{ $C_{e2e,max}$ (in bits/sec/Hz) vs SNR of the three EH protocols for ergodic capacity optimal TS and PS factors.}
	\label{compare_3_ce2e}
	\end{minipage}%
\end{figure}

\section{Conclusion and future work} \label{conclusion}
{In this paper, we derived the distribution of the maximum normalised e2e SNR at the destination node in a CR scenario. Using these results, we characterised the distribution of the maximum e2e SNR RV.} We considered opportunistic selection of the EH relay that maximised the e2e SNR. We demonstrated the viability of a particular choice of normalising constants to characterise this asymptotic distribution of the maximum over a sequence of normalised i.n.i.d. RVs using EVT. Furthermore, we showed the utility of these results in deciding the optimum TS and PS factors, which (i) minimised the outage probability and (ii) maximised the ergodic capacity, at the destination. The solution for this optimisation problem was simplified using results from stochastic ordering.
\par Interesting extensions of this work includes analysis of scenarios where all the available relays forward decoded information to the destination. Another relevant extension of this work would be to incorporate co-channel interference at the relays into the system model. While deteriorating the first hop SNR, co-channel interference at the relay nodes will contribute to the energy harvested. Hence, the trade off between loss in SNR and availability of energy for harvesting in such a scenario will be an interesting study.

\begin{appendices}
 \section{Proof for Theorem \ref{cdf_e2e_thm}} \label{cdf_e2e_proof}
Note that $\gamma_{e2e,\ell}=\gamma_{1,\ell}\min(1,\varphi_{2,\ell})$. Now, let $Y=\min(1,\varphi_{2,\ell})$. Then, the CDF of $Y$ is given by
\begin{equation}
    F_Y(y) =  \begin{cases}
		1, & y \geq 1,  \\
		1-exp(-y\nu_\ell), & 0\leq y \leq1, \\
		0 & \text{o.w}.
		\end{cases}
		\label{cdf_Y}
\end{equation} 
\newline
Hence, the CDF of $\gamma_{e2e,\ell}$ is given by $F_{\gamma_{e2e,\ell}}(\gamma)=\mathbb{P}(\gamma_{1,\ell}Y \leq \gamma)$. Thus, 
\begin{align}
    F_{\gamma_{e2e,\ell}}(\gamma)  = \int\limits_0^\infty F_Y\left(\frac{\gamma}{x_1} \right) f_{\gamma_{1,\ell}}(x_1) \ dx_1 .
\end{align}
Now from (\ref{cdf_Y}), $F_Y\left(\frac{\gamma}{x_1} \right)$ will be unity for all values of $x_1>\gamma$ and hence we can rewrite the previous integral as follows : 
\begin{align}
     F_{\gamma_{e2e,\ell}}(\gamma) & = \int\limits_0^\gamma f_{\gamma_{1,\ell}}(x_1) \ dx_1 + \int\limits_\gamma^\infty(1-\exp(-\nu_\ell  \gamma/x_1) \ f_{\gamma_{1,\ell}}(x_1) \ dx_1 \\
     & = \int\limits_0^\infty  f_{\gamma_{1,\ell}}(x_1) \ dx_1 + \int\limits_\gamma^\infty  \exp(-\nu_\ell  \gamma/x_1) f_{\gamma_{1,\ell}}(x_1) \ dx_1. \label{new}
\end{align}
Since $\gamma_{1,\ell} \sim \text{Exp}(\theta_\ell)$, $ F_{\gamma_{e2e,\ell}}(\gamma)$ can now be written as,
\begin{align}
    F_{\gamma_{e2e,\ell}}(\gamma) & = 1-\theta_\ell \int\limits_\gamma^\infty \exp\left(-\theta_\ell x_1-\frac{\nu_\ell \gamma}{x_1} \right) \ dx_1.
\end{align}

Now, expanding the second exponential, we have

\begin{align}
    F_{\gamma_{e2e,\ell}}(\gamma) & = 1-\theta_\ell  \int\limits_\gamma^\infty \exp(-\theta_\ell x_1) \sum\limits_{k=0}^\infty \frac{(-1)^k}{k!} \left(\frac{\nu_\ell \gamma}{x_1} \right)^k \ dx_1
\end{align}
By Fubini's theorem, we can exchange the integral and the summation in the previous expression we have,
\begin{align}
    F_{\gamma_{e2e,\ell}}(\gamma)  & =  1-\theta_\ell \left \lbrace \sum\limits_{k=0}^\infty  \frac{(-\nu_\ell \gamma)^k}{k!} {\int\limits_\gamma^\infty \exp(-\theta_\ell x_1) \ x_1^{-k} \ dx_1} \right \rbrace .
    \label{b_step7}
\end{align}
Next, by applying the transformation $y=\frac{x_1}{\gamma}$, (\ref{b_step7}) can be rewritten as 
\begin{align}
   F_{\gamma_{e2e,\ell}}(\gamma)   & =  1-\theta_\ell \left \lbrace \sum\limits_{k=0}^\infty  \frac{(-\nu_\ell \gamma)^k}{k!} {\int\limits_{1}^\infty \gamma \exp(-\theta_\ell y \gamma) \ (y \gamma)^{-k} \ dy} \right \rbrace .
    \label{b_step7}
\end{align}
Thus, we have,
\begin{align}
 F_{\gamma_{e2e,\ell}}(\gamma)    =1-\theta_\ell \sum\limits_{k=0}^ \infty \frac{(-\nu_\ell)^k}{k!} \gamma E_k(\theta_\ell \gamma),
    \label{cdf_e2e}
\end{align}
where $E_n(x)$ is the exponential integral function given by $E_n(x) = \int\limits_{1}^\infty \exp(-xt)t^{-n} \ dt.$
\end{appendices}
 
\bibliographystyle{IEEEtran}
\bibliography{reference,bibJournalList}
\end{document}